\documentclass[a4paper,oneside]{memoir}\setsecnumdepth{subsection}\settocdepth{subsection}
\usepackage[utf8]{inputenc}
\usepackage[american,ngerman,brazilian,greek,british]{babel}
\usepackage[T1]{fontenc}
\usepackage[sc,osf]{mathpazo}
\usepackage[scaled=0.86]{berasans}
\usepackage[scaled=1.03]{inconsolata}
\usepackage{mathtools}
\usepackage{layout}
\usepackage{amsmath,amssymb,amsthm}
\usepackage{xspace}
\usepackage{color}
\usepackage[colorlinks,linkcolor=magenta,urlcolor=blue,
]{hyperref}
\usepackage{graphicx}
\usepackage[babel]{microtype}
\usepackage{csquotes} 
\usepackage[backend=bibtex,style=links,maxnames=10,isbn=false,firstinits=true]{biblatex} 
\bibliography{biblio} 

\DeclareMathOperator{\tr}{tr}

\DeclareMathOperator{\sign}{sign}

\let\originalleft\left
\let\originalright\right
\renewcommand{\left}{\mathopen{}\mathclose\bgroup\originalleft}
\renewcommand{\right}{\aftergroup\egroup\originalright}

\newcommand{\ket}[1]{\left| #1 \right\rangle}
\newcommand{\ketpequeno}[1]{| #1 \rangle}
\newcommand{\braket}[2]{\left\langle #1 \middle| #2 \right\rangle}
\newcommand{\ketbra}[2]{\left|#1\middle\rangle\middle\langle#2\right|}
\newcommand{\exval}[3]{\left\langle #1 \middle| #2 \middle| #3 \right\rangle}
\newcommand{\prin}[2]{\left\langle#1, #2\right\rangle}
\newcommand{\norm}[1]{\left\|#1\right\|}

\newcommand{\abs}[1]{\left|#1\right|}
\newcommand{\mean}[1]{\left\langle#1\right\rangle}
\newcommand{\comm}[2]{\left[#1,#2\right]}

\newcommand{\de}[1]{\left(#1\right)}
\newcommand{\De}[1]{\left[#1\right]}
\newcommand{\DE}[1]{\left\{#1\right\}}

\newcommand{\mathand}{\quad\text{and}\quad}
\newcommand{\marginal}{\mathcal{C}}
\newcommand{\chsh}{\mathcal{B}}

\newcommand{\Hi}{\mathcal{H}}

\newcommand{\id}{\mathbb{1}}

\newcommand{\R}{\mathbb{R}}
\newcommand{\Q}{\mathbb{Q}}
\newcommand{\C}{\mathbb{C}}
\newcommand{\N}{\mathbb{N}}
\newcommand{\eff}{\mathcal{E}}
\newcommand{\pro}{\mathcal{P}}
\newcommand{\selfadjoint}{\mathcal{O}}
\newcommand{\state}{\mathcal{D}}
\newcommand{\north}{\mathcal{N}}
\DeclareMathOperator{\dint}{d\!}

\newcommand{\Chi}{\mathrm{X}}

\newcommand{\eg}{\textit{e.g.}\@\xspace}
\newcommand{\ie}{\textit{i.e.}\@\xspace}
\newcommand{\etal}{\textit{et al.}\@\xspace}

\newtheorem{theorem}{Theorem}
\newtheorem{lemma}[theorem]{Lemma}
\newtheorem{definition}[theorem]{Definition}
\newtheorem{corollary}[theorem]{Corollary}
\newtheorem{assumption}{Assumption}
\newtheorem{problem}{Problem}

\hypersetup{pdftitle={Quantum realism and quantum surrealism},pdfauthor={Mateus Araújo}}
\title{Quantum realism and quantum surrealism}
\author{Mateus Araújo}
\date{\today}

\newcommand*{\titleP}{\begingroup%
\raggedright
\vspace*{0.05\textheight}
\hspace*{0.15\textwidth}{\Huge QUANTUM}\\
\vspace*{0.06\textheight}
\hspace*{0.20\textwidth}{\Huge REALISM}\\
\vspace*{0.06\textheight}
\hspace*{0.45\textwidth}{\Huge AND}\\
\vspace*{0.06\textheight}
\hspace*{0.55\textwidth}{\Huge QUANTUM}\\
\vspace*{0.06\textheight}
\hspace*{0.60\textwidth}{\Huge SURREALISM}\\
\vspace*{0.13\textheight}

\begin{center}
\Large Master's thesis\\
\Large Presented to the Graduate Program in Physics of the Universidade Federal de Minas Gerais\footnote{This version incorporates further corrections.}
\end{center}

\vspace*{0.11\textheight}

{\large Author: Mateus Araújo Santos}\\
\vspace*{0.01\textheight}
{\large Supervisor: Marcelo O. Terra Cunha} \\
\vspace*{0.01\textheight}
{\large Examiners: Ernesto F. Galvão} \\
{\large \textcolor{white}{Examiners:} Carlos H. Monken} \\
\vspace*{0.03\textheight}
{\large June, 2012}
\vfill
\endgroup}

\setlrmarginsandblock{1.5in}{1in}{*}
\checkandfixthelayout

\begin{document}

\frontmatter*
\titleP
\thispagestyle{empty}
\cleardoublepage
\thispagestyle{empty}
\vspace*{0.32\textheight}
\epigraph{\ldots we always have had a great deal of difficulty in understanding the world view that quantum mechanics represents. At least I do, because I'm an old enough man that I haven't got to the point that this stuff is obvious to me. Okay, I still get nervous with it. And therefore, some of the younger students\ldots you know how it always is, every new idea, it takes a generation or two until it becomes obvious that there's no real problem. It has not yet become obvious to me that there's no real problem. I cannot define the real problem, therefore I suspect there's no real problem, but I'm not sure there's no real problem.}{Richard Feynman}

\selectlanguage{brazilian}
\chapter*{Agradecimentos}

	À minha Luciana, por ter me feito um homem feliz e por ter conseguido controlar seus ciúmes dessa minha amante.
	
	Aos meus pais, por serem quem são, e por me tornarem quem sou. Seu apoio foi e ainda é indispensável.

	Ao meu orientador Marcelo Terra Cunha, por ter me dado a liberdade de putanejar enquanto eu podia, e por ter me mandado trabalhar quando eu precisava.

	Ao meu grande amigo Marco Túlio Quintino, sem quem essa dissertação seria muito pior.

	Ao Marcelo França, pelas conversas fiadas que me impediam de trabalhar, e por me impedir de ignorar suas sugestões.

	A Gláucia Murta, pela ajuda indispensável em ler e reler a dissertação em busca de erros e passagens obscuras. Qualquer falha de matemática ou de estilo que tenha permanecido no texto é culpa dela. Também agradeço por ser um recurso local capaz de realizar protocolos inacessíveis a uma pessoa altamente não-local.

	Aos meus amigos da Pós, agradeço pelo bom ambiente. Vocês tornam possível ser feliz e aprender física.

	Aos professores da Pós, por tudo o que me ensinaram, e por tudo o que não me ensinaram.

\selectlanguage{british}
\tableofcontents
\chapter*{Abstract}

In this thesis we explore the questions of what should be considered a ``classical'' theory, and which aspects of quantum theory cannot be captured by any theory that respects our intuition of classicality.

This exploration is divided in two parts: in the first we review classical results of the literature, such as the Kochen-Specker theorem, von Neumann's theorem, Gleason's theorem, as well as more recent ideas, such as the distinction between $\psi$-ontic and $\psi$-epistemic ontological models, Spekkens' definition of contextuality, Hardy's ontological excess baggage theorem and the PBR theorem.

The second part is concerned with pinning down what should be the ``correct'' definition of contextuality. We settle down on the definition advocated by Abramsky and Branderburger, motivated by the Fine theorem, and show the connection of this definition with the work of George Boole. This definition allows us to unify the notions of locality and noncontextuality, and use largely the same tools to characterize how quantum mechanics violates these notions of classicality. Exploring this formalism, we find a new family of noncontextuality inequalities. We conclude by reviewing the notion of state-independent contextuality.
\chapter{Introduction}
	\epigraph{Quantum mechanics is magic.}{Daniel Greenberger}

	This thesis is meant to explore the question posed by Chris Fuchs: what is ``Zing!'' \cite{fuchs02}? What is the property of quantum mechanics which is essentially quantum, absent from any classical theory? Contrary to the goals of Chris Fuchs, our exploration is operationalist rather than axiomatic: our ``Zing!'' is not a deep axiom that reveals the essence of quantum theory, but rather logically connected sets of probability distributions that cannot be reproduced by any classical theory. Although finding his axiom would be nice, we feel that our approach is more useful, as these sets of probability distributions are the resources needed for \textit{quantum magic}: quantum computing and quantum key distribution.

	This is emphatically not a historical account of the subject: these are plentiful, and another one is unnecessary. Therefore, we shall try to keep references to the great works of von Neumann, Bell, Kochen, and Specker to a bare minimum, while emphasising the newer\footnote{As a result, the median year of publishing of our references is 2002.} works of Abramsky, Busch, Cabello, Hardy, Pitowsky, and Spekkens. The sole exception shall be the work of George Boole, that although very old is still very unknown.

	Given a general picture of my motivations and goals, let me now give a more detailed account of the structure of this thesis.

	Chapter \ref{cha:onto} presents introductory material\footnote{The reader that is already well-acquainted with the subject (or a mathematician) may find it better to skip it.} on the question ``is quantum mechanics really different from `classical' theories?''. It begins by capturing some notions of classicality within the framework of ontological theories; then this question is made more precise as ``is there an ontological embedding of quantum theory?''. 

	The chapter proceeds by detailing specific ontological models, and showing which problems arise in trying to reproduce the results of quantum mechanics within them. These problems are then understood as their failure to respect noncontextuality, a notion that we argue to be fundamental in defining classicality. After giving a precise definition of noncontextuality, we proceed to prove Spekkens' theorem of the impossibility of embedding quantum theory within a preparation noncontextual ontological model.

	We proceed then to revisit our assumptions, and try to find whether a less ambitious notion of classicality can embed quantum theory. To do that, we revisit the historical theorems of von Neumann and Gleason, culminating with the recent version of Busch. In each of their frameworks, a ``classical'' formulation of quantum mechanics is again ruled out.

	The next stop is the famous theorem of Kochen and Specker, that uses the weakest assumptions yet. We present three recent versions of it, by Cabello \etal, Yu and Oh, and Peres and Mermin, that are considerable simplifications of the original proof.

	The chapter concludes by presenting a recent theorem of Hardy, that ``any ontological embedding of quantum theory is very uncomfortable'', and two specific contextual ontological embeddings of quantum theory.

	Our conclusion is then that any \emph{reasonable} ontological embedding of quantum theory is impossible; therefore there \emph{is} something more in quantum mechanics that classical theories cannot quite capture. Chapter \ref{cha:prob} is then dedicated to detail what this something is. 

	We begin by constructing our final definition of noncontextuality. Based on the recent work of Abramsky and Brandenburger, we show that the Fine theorem admits a natural generalization that applies to any set of observables, without regard to spatial separation. This generalization in its turn motivates a definition of noncontextuality that is a natural generalization of the definition of locality, with mostly the same mathematical structure -- this allows us to consider generalizations of Bell inequalities that test noncontextuality instead of locality. Interestingly, this ``new'' definition was already implicit in the ancient works of Boole (and in the more recent works by Pitowsky), which motivates us to call these generalized Bell inequalities \textit{Boole} inequalities.

	This ``new'' approach is then formalized via a classical problem in mathematics, the marginal problem. Using its formalism, we gain access to powerful tools to separate contextual from noncontextual probability distributions, and with them derive a new result: a set of Boole inequalities that completely describes an infinite family of noncontextual polytopes.

\chapter{Notation and definitions}

	The purpose of this part of the thesis is only to establish notation, not to teach quantum mechanics to anyone. If one needs such an introduction, we recommend the excellent book of Michael Nielsen and Isaac Chuang \cite{chuang00}.

	We say that an operator $A$ is self-adjoint, \ie, $A = A^*$, if $\braket{\phi}{A\psi} = \braket{A\phi}{\psi} = \exval{\phi}{A}{\psi}$ for all $\ket{\phi},\ket{\psi}$. We shall only deal with finite-dimensional operators. The set of all self-ajoint operators is $\selfadjoint(\Hi)$.

	A quantum-mechanical observable is a self-adjoint operator.

	We say that an operator $A$ is positive, \ie, $A\ge0$, if $\exval{\psi}{A}{\psi}\ge 0$ for all $\ket{\psi}$. 

	A quantum state $\rho$ is a positive operator such that $0 \le \tr \rho \le 1$ \cite{hardy01}. Since we shall have no use for states such that $\tr\rho < 1$, we can omit the normalization of our quantum states without ambiguity. The set of all quantum states is $\state(\Hi)$. A pure quantum state is an extremal point of $\state(\Hi)$, a rank-one projector $\psi$. The vector of a pure quantum state will be denoted by $\ket{\psi}$, and the vectors are connected to the projectors by \[\psi = \ketbra{\psi}{\psi}.\]
	The set of all pure states is $\pro\Hi$.

	An effect $E$ is a positive operator smaller than identity, \ie, $0 \le E \le \id$. The set of all effects is $\eff(\Hi)$. A set of effects $\{E_i\}$ such that $\sum_i E_i = \id$ describes a measurement\footnote{Except for the post-measurement state.} and is called a POVM.

	A projector $\Pi$ is a self-adjoint operator such that $\Pi^2 = \Pi$. The set of all projectors is $\pro(\Hi)$. A set of projectors $\{\Pi_i\}$ such that $\sum_i\Pi_i = \id$ describes a measurement and is called a PVM. Note that a PVM is a special case of a POVM.

	The Born rule is the quantum mechanical rule for associating measurement probabilities with states and effects. We say that
	\[ p(i|\rho,E) = \tr \rho E_i. \]

\mainmatter*
\chapter{Ontological embeddings of quantum theory}\label{cha:onto}
\setlength{\epigraphwidth}{1.4\epigraphwidth}
	\epigraph{Classical measurements reveal information. Quantum measurements produce information.}{Marcelo Terra Cunha}
\setlength{\epigraphwidth}{0.8\epigraphwidth}

	The quest for embedding quantum mechanics in a ``classical'' theory is almost as old as quantum theory itself. People were disturbed with the role of measurement in the theory, particularly with its intrinsic randomness and non-repeatability. So they tried to explain away these features as emergent, rather than fundamental, as if they appeared because of a lack of control and understanding of a more refined theory, that would describe the ``deeper'' physics behind quantum phenomena. We call this refined theory an \emph{ontological} theory.

	But despite being familiar, the words ``classical'' and ``ontological'' have very fuzzy meanings. In the next section we shall pin them down and clarify them.

\section{What is an ontological theory?}

	The first ontological models that appeared tried to ``solve'' the problem of non-determinism. They postulated that $\psi$ was not the real state of nature, but rather some kind of shadow of it. So they postulated that there \emph{was} a real state, an ontic state\footnote{The reader that is well-acquainted with the subject might be wondering when the expression ``hidden-variable'' will appear. Well, it won't.}, called $\lambda$, that if known would render all measurement outcomes deterministic. That is, given a PVM\footnote{Even the most determined determinist can't hope for a POVM to be deterministic. We'll explain why in a while.} $M = \DE{M_k}$, the probability of outcome $k$ given $\lambda$ would be either $0$ or $1$, that is, we can define a response function \[\xi_{k|M} : \Lambda \to \DE{0,1},\] such that $\xi_{k|M}(\lambda)$ is the probability of outcome $k$. Here, $\Lambda$ is any space in which our ontic states $\lambda$ are defined, and to account for the fact that $\sum_k M_k = \id$, we require that $\sum_k \xi_{k|M}(\lambda) = 1$ for all $\lambda$. This is just the requirement that some outcome must occur in a measurement.

	Then the \emph{subjective} indeterminism of quantum theory would be recovered by the ignorance of which ontic states were really present in a experiment. That is, a quantum state $\psi$ would determine a probability distribution $\mu_\psi(\lambda)$ over $\Lambda$. This property can be thought of as ``you were trying to generate state $\psi$, but you ended up generating an ensemble of ontic states $\mu_\psi(\lambda)$''. As in quantum (and classical) mechanics, we shall call the ensemble $\mu_\psi(\lambda)$ itself a state, while reserving the term pure ontic state for the individual $\lambda$, which can of course be represented as an ensemble with a $\delta$ distribution.

	Of course, we want this subjective indeterminism to agree with the predictions of quantum mechanics, so 
	\begin{equation}\label{eq:recoverqmdet}
	p(k|\psi,M)= \int_\Lambda\dint\lambda\, \mu_\psi(\lambda) \xi_{k|M}(\lambda) = \tr\psi M_k.
	\end{equation}

\subsection{On mixed states and POVMs}

	The early literature of ontological theories did not do this separation between states and measurements\footnote{With the honourable exception of the Kochen-Specker model, discussed in section \ref{sec:kochenspeckermodel}.} \cite{bell66,mermin93}; instead they tried to define a deterministic value function $v(M_k,\psi,\lambda)$ that would answer with certainty the outcome of an experiment, given the quantum state and the ontic state, and recover the quantum statistics by averaging over $\lambda$. This is quite problematic, since it can only describe models in which $\psi$ itself has an ontic status\footnote{See section \ref{sec:psiontic} for further discussion of this point.}; it therefore can never describe experiments where the quantum state is explicitly epistemic, \eg, a mixed state. For instance, let's say we have two pure states $\psi$ and $\phi$ with different deterministic outcomes $v(M_k,\psi,\lambda)$ and $v(M_k,\phi,\lambda)$. Then if I prepare state $\psi$ with probability $p$ or state $\phi$ with probability $(1-p)$, corresponding to the mixed state $\rho = p \psi + (1-p)\phi$, the outcome must be
	\[ v(M_k,\rho,\lambda) = pv(M_k,\psi,\lambda)+ (1-p)v(M_k,\phi,\lambda), \]
	which is neither $0$ nor $1$ for non-trivial $p$, a contradiction.

	Using probability distributions like we do, this can be accommodated in a very natural manner:
	\begin{lemma}\label{lem:mixedontic}
	If one prepares the quantum states $\psi_i$ with probabilities $p_i$, then the corresponding ontic state is \[\mu_{(p_i,\psi_i)}(\lambda) = \sum_i p_i \mu_{\psi_i}(\lambda) \]
	\end{lemma}
	\begin{proof}
	Quantum mechanics tells us that $p(k|(p_i,\psi_i),M) = \sum_i p_i p(k|\psi_i,M)$. Writing these probabilities ontologically, we have\footnote{When doing calculations we shall often omit the integration variable $\lambda$, but only when there's no risk of ambiguity.}
	\[ \int_\Lambda \mu_{(p_i,\psi_i)} \xi_{k|M}  = \sum_i p_i \int_\Lambda \mu_{\psi_i}\xi_{k|M}.\]
	Since $\xi_{k|M}$ is positive and arbitrary, this implies that \[\mu_{(p_i,\psi_i)}(\lambda) = \sum_i p_i \mu_{\psi_i}(\lambda). \]
	\end{proof}
	Note that this same rule is used to describe convex combinations of states in quantum and classical mechanics.

	The issue with POVMs is similar: one can implement the POVM \[E = \DE{p\ketbra{0}{0},p\ketbra{1}{1},(1-p)\ketbra{+}{+},(1-p)\ketbra{-}{-}} \] simply by measuring the PVM $M = \DE{\ketbra{0}{0},\ketbra{1}{1}}$ with probability $p$ and the PVM $N = \DE{\ketbra{+}{+},\ketbra{-}{-}}$ with probability $1-p$ \cite{ariano05}; we must have then $\xi_{0|E}(\lambda) = p\xi_{0|M}(\lambda)$, which is obviously not deterministic. We must accept, then, that for these kinds of ``mixed'' POVMs\footnote{Following \cite{ariano05}, we are calling ``mixed'' the POVMs that can be written as a convex combination of different POVMs, and ``pure'' those who can't.} the response functions must be modified to
	\[ \xi_{k|E} : \Lambda \to [0,1], \]
	that is, allowing the whole interval $[0,1]$ as image.

	For ``pure'' POVMs, this argument does not apply, and we can not decide \textit{a priori} whether to demand them to be deterministic. In fact, it is fruitful to allow even PVMs to be objectively non-deterministic\footnote{However discomforting that may seem for some people, it's certainly a milder discomfort than abandoning the notion of reality altogether as in quantum mechanics. See section \ref{sec:naive}.}, so we shall not exclude this possibility.

	The most general case is, therefore,
	\begin{equation}\label{eq:recoverqmrandom}
	p(k|\rho,E) = \int_\Lambda\dint\lambda \mu_\rho(\lambda) \xi_{k|E}(\lambda) = \tr\rho E_k,
	\end{equation}
	and this is what an ontological theory should strive to reproduce, only falling back to pure states and PVMs when unavoidable.

\section{Ontological models}

	With the definitions given in the previous section, it is already possible to construct some examples of ontological theories, to examine their features in a more concrete manner.

\subsection{The naïve ontology}\label{sec:naive}

	If we allow an ontological model to have objective non-determinism, what we gain in relation to quantum mechanics? Not much, actually. This ontological model is so similar to quantum mechanics that it can be confounded with a naïve interpretation of it, that ascribes ontological status to the pure states. Nevertheless, it is quite useful to examine meticulously this ontological model, to be aware of the problems that such a naïve interpretation has. This particular model was first proposed by \cite{beltrametti95}, and further explored in \cite{harrigan10}.

	In this model, we are considering the pure states $\psi$ to be the ontic states $\lambda$, so we identify the ontic state space $\Lambda$ with $\pro\Hi$, and define
	\[ \mu_\psi(\lambda) = \delta(\lambda-\psi). \]
	The response function is then
	\[ \xi_{k|E}(\lambda) = \tr \lambda E_k,\]
	and we recover the results of quantum mechanics by 
	\[ p(k|\psi,E) = \int_\Lambda \dint\lambda\, \delta(\lambda-\psi)\tr \lambda E_k = \tr \psi E_k. \]
	We can see, then, that mathematically this ontological model is quite trivial. One interesting thing to examine, though, is the representation of mixed states in this formalism. Following lemma \ref{lem:mixedontic}, we see that
	\[ \rho = \sum_i p_i \psi_i \quad\mapsto\quad \mu_\rho(\lambda) = \sum_i p_i \delta(\lambda-\psi_i), \]
	which trivially reproduces the required quantum statistics. The problem with this approach, however, is that the ontic state $\mu_\rho(\lambda)$ depends on which convex decomposition of $\rho$ we chose to use. This makes the the notation $\mu_\rho$ suspect, since it should actually be $\mu_{(p_i,\psi_i)}$, and blatantly violates the $C^*$-algebraic definition of state \cite{strocchi12}, that requires that states that gives rises to the same statistics to have the same mathematical representation. We call this (unwanted) feature preparation contextuality, which we shall define more carefully in section \ref{sec:spekkenscontextuality}.

	Remember that it is common for beginners to be surprised by the fact that it is impossible to know which convex combination was actually used to construct a given density matrix. Regarding the pure states as ontological, this feeling becomes quite natural, since the mystery is why should the state $\mu_{(p_i,\psi_i)}$ give the same statistics as the state $\mu_{(q_i,\phi_i)}$ when $\sum_i p_i \psi_i = \sum_i q_i \phi_i$.

	To solve this problem, one might be tempted to ignore common sense (and lemma \ref{lem:mixedontic}) and ascribe ontological status to mixed states, identifying $\Lambda$ with $\state(\Hi)$ instead of $\pro\Hi$; then the ontic states would be just
	\[ \mu_\rho(\lambda) = \delta(\lambda-\rho), \]
	relieving us of the basis-dependence. But this is in fact a terrible idea, since one can always write a mixed state $\rho$ as a convex combination of two different states $\sigma_0$ and $\sigma_1$, as 
	\[\rho = p\sigma_0 + (1-p)\sigma_1.\]
	If you want to regard every mixed state as ontological, you have, by lemma \ref{lem:mixedontic},
	\[\delta(\lambda-\rho) = p\delta(\lambda-\sigma_0) + (1-p)\delta(\lambda-\sigma_1),\]
	a flat-out contradiction.

	One can now begin to suspect that it is not possible to avoid preparation contextuality; this will be proved in section \ref{sec:preparationcontextuality}. For now, we see that even the most humble ontological model, that does not even provide determinism, already has some very undesirable features. It would be a question then if a deterministic ontological model is even possible; fortunately this question was answered a long time ago in the positive. We shall see how in the next subsection.

\subsection{Constructing a deterministic ontological model}\label{sec:bellmeu}

	In 1964, Bell had an idea on how to make a deterministic ontological model \cite{bell66}: hide the quantum mechanical probability of an outcome in the measure of the set of ontic states associated to that outcome. I shall present here a modified version of his model that makes this point quite clear. 

	This model can describe in a deterministic way the measurement of a one-qubit PVM $\Pi = \DE{\Pi_0,\Pi_1}$. The ontic space is $\Lambda = \pro\Hi \times [0,1]$, with ontic variable $\lambda = (\lambda_\psi,\lambda_x)$. The ontic state of a given quantum state $\psi$ is
	\[ \mu_\psi(\lambda_\psi,\lambda_x) = \delta(\lambda_\psi-\psi),\]
	and the response functions\footnote{Note that the response functions depend explicitly on the label of the projectors, so it would be desirable to set a consistent ordering convention to avoid giving different results to $\DE{\ketbra{0}{0},\ketbra{1}{1}}$ and $\DE{\ketbra{1}{1},\ketbra{0}{0}}$.} are
	\begin{align*}
	\xi_{0|\Pi}(\lambda_\psi,\lambda_x) &= \Theta(\tr\lambda_\psi \Pi_0-\lambda_x) \\
	\xi_{1|\Pi}(\lambda_\psi,\lambda_x) &= 1-\xi_{0|\Pi}(\lambda_\psi,\lambda_x),
	\end{align*}	
	where $\Theta$ is the Heaviside step function defined by
	\[ \Theta(x) = \begin{cases} 1&\text{ if }x \ge 0, \\
				     0&\text{ if }x < 0.\end{cases} \]
	One then recovers quantum statistics by uniform averaging over the ontic space:
	\begin{align*}
	p(0|\psi,\Pi)  &= \int_\Lambda \mu_\psi \xi_{0|\Pi}\\
		   &= \int_\Lambda \dint \lambda_\psi \dint\lambda_x \, \delta(\lambda_\psi-\psi)\Theta(\tr\lambda_\psi \Pi_0-\lambda_x) \\
		   &= \int_0^1 \dint\lambda_x\, \Theta(\tr\psi \Pi_0-\lambda_x) \\
		   &= \int_0^{\tr\psi \Pi_0} \dint\lambda_x = \tr\psi \Pi_0
	\end{align*}
	The reader might have noticed that although the model claims to only work for a qubit, the mathematical formalism does not make any reference to this, and one might be tempted to think that it actually works for any two-outcome PVM. The fact that it does not work is more subtle, and we shall see why in section \ref{sec:kochenspecker}.


\section{\texorpdfstring{$\psi$-ontic and $\psi$-epistemic models}{psi-ontic and psi-epistemic models}}\label{sec:psiontic}

	Both models presented in the previous section share a common feature: the quantum state has an ontological status. Either the ontic state is the quantum state itself, like in the naïve model, or it is the quantum state supplemented by real number in the unit interval, as in the Bell model. In both cases, knowing the (pure) ontic state $\lambda$ of the system is enough to determine uniquely the (pure) quantum state that was prepared. These kind of models are called\footnote{The concept of ontic and epistemic states was first introduced in \cite{hardy04}, and further formalized in \cite{spekkens05,harrigan10}. A nice discussion of these concepts can be found in \cite{leifer11}.} $\psi$-ontic, and have the equivalent but more operational definition:
	\begin{definition}
	An ontological model is $\psi$-ontic if for different quantum states $\phi$ and $\psi$ the ontic states have disjoint support, \ie,
	\[ \phi \neq \psi \quad\Rightarrow\quad \mu_\phi(\lambda)\mu_\psi(\lambda) = 0 \quad \forall \lambda \]
	\end{definition}
	To motivate this definition it might be useful to make an analogy with classical mechanics: in it, an ontic state is a point in phase space, and ontic properties of it (like energy, momentum) are functions of the phase space point. Likewise, anything that is uniquely determined by the ontic state in an ontological theory should be regarded as ontic itself, as a change in it requires a change of the underlying ontic states. As the quantum state is uniquely determined by the ontic state in $\psi$-ontic models, it has to be regarded as ontic, as it is not possible to change it without changing the underlying ontic states.

	Apart from conceptual clarity, a reason to make this definition is that it is easy to see that $\psi$-ontic models necessarily require instant transfer of information\footnote{Only in the formalism, of course; if they displayed an observable violation of causality that would be a contradiction with quantum mechanics.}. In the first case, where $\psi$ is the whole ontic state, it suffices to consider a measurement in an entangled state: Alice and Bob share $\ket{\phi_+} = \ket{00} + \ket{11}$ and are spatially separated, Alice then measures the PVM $\DE{\ketbra{0}{0},\ketbra{1}{1}}$ and obtains, \eg, the result $0$. Bob's state then changes instantly from $\id$ to $\ket{0}$, violating causality. Of course, if $\psi$ is not the whole ontic state, there is no need for a violation of causality: $\lambda$ can tell us that the state of Bob's system actually was $\ket{0}$ all along, and so the ontic state does not change during the measurement.

	To deal with this case, we need the \textsc{epr} \textit{gedankenexperiment}\footnote{The version presented here is Einstein's version, reproduced in \cite{harrigan10}.} \cite{einstein35}: consider that Alice can also measure the PVM $\DE{\ketbra{+}{+},\ketbra{-}{-}}$; then after her measurement Bob's state will belong to the set $\DE{\ket{0},\ket{1}}$ if she measures the first PVM, or to the set $\DE{\ket{+},\ket{-}}$ if Alice measures the second PVM. Even if the results of any given measurement can be predetermined by $\lambda$, it cannot tell \emph{which} measurement was made\footnote{Indeed, it could conceivably determine which measurement Alice \emph{will} make -- here we are using the assumption that she has \textit{free will}.}. Since Bob's quantum state does depend on which measurement was made (since the four possibilities are different), the formalism needs again instant transfer of information.

	Another way to avoid the violation of causality is to say that $\psi$ is not ontic, but merely the representation of Alice's knowledge of reality, \ie, epistemic. Then what changed after the measurement was actually just what Alice knew about Bob's state, which is in fact a quite reasonable proposition. But this amounts to give up $\psi$-ontic models in favour of $\psi$-epistemic ones\footnote{It is interesting to notice that although we've known this since 1935, the first ontological models were all $\psi$-ontic.}:

	\begin{definition}
	An ontological model is $\psi$-epistemic if it is not $\psi$-ontic.
	\end{definition}

	Again, an analogy with classical mechanics might be useful: the classical mixed state is a probability distribution over the phase space, and it is interpreted as epistemic, as it is merely an ignorance about which is the real phase space point that the system occupies. This is only possible as there is no restriction about the overlaps of different mixed states, \ie, the same phase space point can belong to numerous different mixed states. Notice that this definition is quite weak compared to the classical case: it only requires that there is one pair $\phi$, $\psi$ whose ontic states $\mu_\phi$ and $\mu_\psi$ share a single $\lambda$ in their support.

	The obvious question to ask: is there a $\psi$-epistemic model?

\subsection{The Kochen-Specker model}\label{sec:kochenspeckermodel}

	Even before this question was raised, it was already answered by Simon Kochen and Ernst Specker \cite{kochen67}, by the ontological model they constructed as a counterexample to von Neumann's theorem \cite{vonneumann32}. It seems that the authors were trying to make a model that was somewhat physically plausible, and ended up making a $\psi$-epistemic model. We presented it here as rendered in \cite{harrigan10}.

	The ontic space $\Lambda$ is the unit sphere $S^2$, and we shall use the Bloch vectors $\hat{\psi}$ and $\hat{\phi}$ to represent a pure state $\psi$ and a measurement projector $\phi$ in $S^2$ as well, defined via the isomorphism $\psi = \frac{1}{2}(\id+\hat{\psi}\cdot\sigma)$. The ontic state is then 
	\[ \mu_\psi(\lambda) = \frac{1}{\pi}\Theta(\hat{\psi}\cdot\lambda)\hat{\psi}\cdot\lambda,\]
	making the model clearly $\psi$-epistemic, since the only states that do not overlap are orthogonal states. The response function is given by
	\[ \xi_\phi(\lambda) = \Theta(\hat{\phi}\cdot\lambda). \]
	To recover the quantum statistics, notice that each of $\mu_\psi$ and $\xi_\phi$ has as support an hemisphere centred in $\hat{\psi}$ and $\hat{\phi}$, so their intersection defines a spherical lune. To take advantage of this, let's choose coordinates such that $\hat{\psi}$ and $\hat{\phi}$ lie in the equator of $S^2$, so that $\hat{\psi} = (\cos\psi,\sin\psi,0)$, $\hat{\phi} = (\cos\phi,\sin\phi,0)$, and $\lambda = (\sin\theta\cos\varphi,\sin\theta\sin\varphi,\cos\theta)$. We have then
	\begin{align*}
	p(\phi|\psi) &= \int_\Lambda \dint \lambda\, \frac{1}{\pi}\Theta(\hat{\psi}\cdot\lambda)\hat{\psi}\cdot\lambda \Theta(\hat{\phi}\cdot\lambda) \\
		     &= \frac{1}{\pi}\int_{S^2} \dint\Omega\, \Theta(\sin\theta\cos(\varphi-\psi)) \sin\theta\cos(\varphi-\psi) \Theta(\sin\theta\cos(\varphi-\phi))\\
		     &= \frac{1}{\pi}\int_0^\pi \dint \theta\, \sin^2\theta \int_0^{2\pi}\dint\varphi\, \Theta(\cos(\varphi-\psi))\cos(\varphi-\psi) \Theta(\cos(\varphi-\phi))\\
		     &= \frac{1}{2} \int_{\phi-\frac{\pi}{2}}^{\psi+\frac{\pi}{2}}\dint\varphi\,\cos(\varphi-\psi) \\
		     &= \frac{1}{2} \de{1-\sin(\phi-\psi-\pi/2)} \\
		     &= \frac{1}{2} \de{1+\cos(\phi-\psi)} \\
			 &= \tr \psi\phi.
	\end{align*}

	This model does seem to be the most ``natural'' of the ontological models yet considered, and there have even been attempts to understand it physically \cite{rudolph06}. In this same article, Terry Rudolph explores extensions of the Kochen-Specker model to higher dimensions, but fails to precisely reproduce quantum mechanics with them. A $\psi$-epistemic model for higher dimensions has since then been found (we discuss it in section \ref{sec:ljbr}), but it does not have the simplicity of the Kochen-Specker model, and so it would be unfair to call it an extension of it.

\subsection{\texorpdfstring{Two theorems on $\psi$-epistemic models}{Two theorems on psi-epistemic models}}

	We can see, then, that $\psi$-epistemic models are desirable and can actually be constructed. There are, however, two theorems that say that any such model, if it exists, has to be very unnatural. They are both based on the following idea:
	\begin{lemma}\label{lem:protopbr}
	If there are quantum states $\psi_i$ and measurements $E_i$ such that $\tr\psi_i E_i = 0$ $\forall i$, then there can be no $\lambda_0$ in the support of all $\mu_{\psi_i}$.
	\end{lemma}
	\begin{proof}
	If these conditions are satisfied, then it must be true that
	\[ \int_\Lambda \dint\lambda\,\mu_{\psi_i}(\lambda)\xi_{i|E}(\lambda) = 0,\]
	and therefore that $\xi_{i|E}(\lambda) = 0$ for all $\lambda$ in the support of $\mu_{\psi_i}$. If there is a $\lambda_0$ in the support of all the $\mu_{\psi_i}$, making the model $\psi$-epistemic, then $\sum_i \xi_{i|E}(\lambda_0) = 0$, an absurd, since in the definition of the response functions we require that $\sum_i \xi_{i|E}(\lambda) = 1$ for all $\lambda$.
	\end{proof}
	Of course, if we could prove that for \emph{any} pair of states the hypothesis of the lemma are satisfied, we would have proven that no $\psi$-epistemic model is possible; but for a pair of states the hypothesis of the lemma are satisfied only if they are orthogonal, and by lemma \ref{lem:orthogonality} they must have disjoint support anyway:
	\begin{lemma}
	If there are quantum states $\psi_0,\psi_1$ and measurements $E,\id - E$ such that $\tr\psi_0 E = \tr\psi_1 (\id - E) = 0$, then $\psi_0\psi_1 = 0$
	\end{lemma}
	\begin{proof}
	$\tr\psi_1 (\id -E) = 0\quad\Rightarrow\quad\tr\psi_1 E = 1$, so the support of $\psi_1$ is contained in the support of $E$. But $\tr \psi_0 E=0$ implies that the supports of $\psi_0$ and $E$ are disjoint, and therefore the supports of $\psi_0$ and $\psi_1$ are disjoint, so $\psi_0\psi_1=0$
	\end{proof}
	Instead, the two theorems we shall present consider larger families: the first considers families of three states to show that there are non-trivial examples, and the second argues that the existence of some specific families implies that any $\psi$-epistemic model must be very unnatural.

	\begin{theorem}[Caves, Fuchs, Shack \cite{caves02}]
	If the convex hull of a family of states $\psi_i$ contains $\id/d$, where $d$ is the Hilbert space dimension, then there can be no $\lambda_0$ in the common support of all $\mu_{\psi_i}$.
	\end{theorem}
	\begin{proof}
	For any state $\psi_i$, it is true that $\tr \psi_i (\id-\psi_i) = 0$. If we can find coefficients $\alpha_i$ such that $\DE{\alpha_i(\id-\psi_i)}$ is a POVM, then lemma \ref{lem:protopbr} applies and we're done. What we need is \[\sum_i \alpha_i(\id-\psi_i) = \id,\] for $\alpha_i \ge 0$. Taking the trace on both sides we get that $\sum_i \alpha_i = \frac{d}{d-1}$. Simple algebra then shows us that \[\sum_i \frac{d-1}{d}\alpha_i\psi_i = \frac{1}{d}\id.\]
	\end{proof}

	This theorem was first proven in \cite{caves02}, with a different objective. While it does not exclude $\psi$-epistemic models, it shows there are a wide variety of families of states that can't have an overlap. If the number of states is three, there are already examples in any dimension where they are not orthogonal; see equations \eqref{eq:spekkensstates} for an example.

	The next theorem needs the following (very natural, in the author's opinion) assumption about the composition of different systems:
	\begin{assumption}\label{ass:pbr}
	If two quantum states $\phi$ and $\psi$ are prepared independently, such that their joint state is $\phi\otimes \psi$, then the corresponding ontic state for the joint system is $\mu_{\phi\otimes\psi}(\lambda_A,\lambda_B) = \mu_\phi(\lambda_A)\mu_\psi(\lambda_B)$.
	\end{assumption}
	
	\begin{theorem}[Pusey, Barret, Rudolph \cite{pusey11}]\label{teo:pbr}
	Given assumption \ref{ass:pbr}, no $\psi$-epistemic ontological model of quantum mechanics is possible.
	\end{theorem}
	\begin{proof}
	Consider the four quantum states $\phi_0 \otimes \phi_0$, $\phi_0 \otimes \phi_1$, $\phi_1 \otimes \phi_0$, and $\phi_1 \otimes \phi_1$. If there is a $\lambda_0$ in the support of $\mu_{\phi_0}$ and $\mu_{\phi_1}$, then $(\lambda_0,\lambda_0)$ is in the support of all four $\mu_{\phi_i}(\lambda')\mu_{\phi_j}(\lambda'')$. If there is a POVM $\DE{E_{ij}}$ such that $\tr \phi_i\otimes \phi_j E_{ij} = 0$, then lemma \ref{lem:protopbr} applies and we're done.

	Consider now the particular case $\ket{\phi_0} = \ket{0}$ and $\ket{\phi_1} = \ket{+}$. Then if $E_{ij}$ is the projector onto $\ket{E_{ij}} = \ketpequeno{\phi_i\phi_j^\perp}+\ket{\phi_i^\perp\phi_j}$, it is easy to see that \[\tr \phi_i\otimes \phi_j E_{ij} = \braket{\phi_i\phi_j}{E_{ij}} = 0,\]
	and it is also easy (but tedious) to check that $\sum_{ij} E_{ij} = \id$. Unfortunately, this simple strategy only works for this pair of states, and states with smaller overlap require measurements on a larger number of parts. For the proof of the general case, see the original article\footnote{This proof uses the notation from \cite{leifer11}, which is clearer than the one in the original article.} \cite{pusey11}.
	\end{proof}

	This theorem has two immediate corollaries:
	\begin{corollary}
	Any ontological model of quantum mechanics must violate causality.
	\end{corollary}
	One only has to notice that since the theorem excludes $\psi$-epistemic models, we're left with $\psi$-ontic ones. And we have shown that those violate causality in the beginning of this section.
	\begin{corollary}
	The ontic state space $\Lambda$ is uncountable.
	\end{corollary}
	In a $\psi$-ontic model there is an injection of $\pro(\Hi)$ onto $\Lambda$. Since $\pro(\Hi)$ is uncountable, $\Lambda$ must be uncountable. In fact, even if without assumption \ref{ass:pbr} we can still prove that $\Lambda$ is infinite; we shall do this in section \ref{sec:excessbaggage}.

	The obvious question that this theorem raises is: can we do away with assumption \ref{ass:pbr} and prove once and for all that $\psi$-epistemic models are always impossible? The existence of the Kochen-Specker model already hints that at least some weaker assumption is needed, since it is a \textit{bona fide} $\psi$-epistemic model. Of course, its existence does not contradict the theorem, since it only forbids models for dimension 4 or greater. In fact, soon after the Pusey-Barret-Rudolph was published, some of the same authors showed that without assumption \ref{ass:pbr} they could make a $\psi$-epistemic model for a quantum system of any dimension. We shall describe this model in section \ref{sec:ljbr}.

	This theorem already hints of a theme that shall be recurrent in the search for ontological models: we can in fact make ontological models for quantum theory, and in fact we can make them almost in any way that we like, but there's a price to pay: the various aspects of the model become more and more intertwined. We can't really talk of independent quantum systems, separation between state and experiment, nor even (as we shall see in the next section) talk about a measurement outcome without talking about the whole experiment. Of course, this bodes very badly for the idea of ontological models: in the extreme limit of this interdependence our ontological model only lists possible experiments and their results, without ever trying to make sense of them in a simpler and more general theory. A model like this wouldn't be falsifiable by its very nature, but precisely because of this it is a perversion of the scientific method \cite{popper34}, and should therefore be rejected on methodological grounds.

	What we seek, therefore, is not \emph{any} ontological model, but one that might have some plausibleness. The ontological models present hitherto are of course very contrived, but by themselves they should not be taken as an evidence against the possibility of a reasonable ontological model, since they were conceived only as proofs of principle, without any inspiration from physical grounds.

\section{Contextuality}\label{sec:spekkenscontextuality}

	One should contrast the state of research into contextuality to the state of research into nonlocality. It is quite clear that nonlocality has a better status: it was subjected to experimental tests much earlier\footnote{1972 \cite{freedman72}, in contrast with 2000 \cite{michler00}.}, and also had its potential as a resource for practical applications recognized much earlier\footnote{1991 \cite{ekert91}, \textit{versus} 2000 \cite{bechmann-pasquinucci00}.}

	This state of affairs has many causes, which certainly includes the intuitive appeal of nonlocality via its relation with relativity, but I'd like to focus in a more formal one: the definitions of nonlocality and contextuality. Right in the first paper about nonlocality, John Bell \cite{bell64} already gave a clear operational definition of nonlocality, that was not dependent on quantum theory, but instead only on a general probabilistic framework. By contrast, the first definition of contextuality, also due to John Bell\footnote{The concept appeared first in 1966 \cite{bell66}, in a critique of the Gleason theorem, whereas the name ``contextuality'' was created in 1978 \cite{clauser78}, by Clauser and Shimony.}, was very specific to quantum theory, and was not at all operational:

	\begin{definition}[Bell's contextuality]\label{def:bellcontext}
	We say that an ontological model for quantum theory is noncontextual if the response function associated to the outcome $k$ of a PVM $M = \DE{\Pi_k}$, \ie, $\xi_{k|M}(\lambda)$ depends only on $\Pi_k$ and not on the whole $M$.
	\end{definition}

	This definition also lacks conceptual clarity: John Bell even thought that it was reasonable for a physical theory to be contextual \cite{bell66}: 
	\begin{quote}
	The result of an observation may reasonably depend not only on the state of the system (including hidden variable) but also on the complete disposition of the apparatus.
	\end{quote}

	But one consequence of contextuality is precisely the violation of causality that he abhorred: consider, for instance, the PVM \[M = \{{\Pi_0\otimes\id},{\Pi_1\otimes\id},{\id\otimes\Pi_0},\\{\id\otimes\Pi_1}\}.\] If the real result $\xi_{0|M}(\lambda)$, associated with the projector ${\Pi_0\otimes \id}$, depends on whether the other side of the PVM is $\id\otimes\Pi_0,\id\otimes\Pi_1$ or ${\id\otimes\Pi'_0},{\id\otimes\Pi'_1}$, then the apparatuses must always be able to communicate their arrangement to each other, even when the choice of arrangement is made with a space-like separation, which is of course absurd. This settles the question about ontological models of independent quantum systems. But what about single systems? Is there any unacceptable consequence of contextuality for them?

	Yes! It also implies on a violation of causality. As put by Asher Peres and Amiran Ron \cite{peres88}:
	\begin{quote}
	More generally, if $\comm{A}{B} = \comm{A}{C} = 0$ but $\comm{B}{C} \neq 0$, suppose that we measure $A$ first and only a later time decide whether to measure $B$ or $C$ or none of them. How can the outcome of the measurement $A$ depend on this future decision?
	\end{quote}

	Furthermore, this whole story about communicating apparatuses is quite queer, even when it is not a violation of causality. After all, all the evidence we have is that the measurement of commuting observables does not affect each other, and an ontological theory that requires this kind of communication would be very weird indeed. Another problem is that this communication could affect only the individual measurements $\xi_{i|M}(\lambda)$, and must never be detectable in the quantum experiments we do. To postulate this kind of ``cryptocontextuality''\footnote{With apologies to Asher Peres.} seems very unscientific: we would be making a theory which is about precisely what we \emph{can't} measure.

	Another way to think about the weirdness of a contextual model is operationally: imagine that you are an experimentalist that has implemented an apparatus that can differentiate between the ground state and the excited states of a many-level atom. You try it hard, repeat your experiment a lot of times, with different input states, gather the statistics, and is confident that your apparatus is quite trustworthy; you now want to teach a friend experimentalist how to build a similar apparatus. Quite simple, isn't it? You just tell him how you did, ask him to gather statistics, and compare with yours: if the statistics match, you've implemented the same experiment. Except it isn't so if your physical theory is contextual: the statistics of the projector $\Pi_0$ (the projector onto the ground state) are not enough to determine the results of the experiment, since according to definition \ref{def:bellcontext} the real results $\xi_{0|\Pi}(\lambda)$ depend on the rest of the (unmeasured) projectors; and these are not only the higher energy levels of the atom, but can in principle include any environmental data, such as the apparatus' mass, the local weather, whether Virgo is ascendant\ldots

	In this way, we are rendered incapable of comparing experiments and establishing patterns, the very foundation of our scientific method. Notice the strong parallel between this discussion and the definitions of state and observable in the $C*$-algebraic axiomatization done by Franco Strocchi \cite{strocchi12}. This motivates a new definition of contextuality, due to Spekkens \cite{spekkens05}, that takes into account these arguments:

	\begin{quote}
	A noncontextual ontological model of an operational theory is one wherein if two experimental procedures are operationally equivalent, then they have equivalent representations in the ontological model.
	\end{quote}

	Within this reasoning, it becomes sufficient to have equivalent statistics to be able to identify different experiments, and we are able again to do science. But a definition that uses only words is quite imprecise, and we should codify it in order to avoid misinterpretations:

	\begin{definition}[Spekkens' contextuality]\label{def:spekkenscontext}
	Let $p(k|P,M)$ be the probability of obtaining the outcome $k$ when doing the measurement $M$ on a state prepared via procedure $P$. Then we say that an ontological model of an operational theory is measurement noncontextual if 
	\begin{equation} p(k|P,M) = p(k|P,M') \quad\forall P\quad\Rightarrow\quad M=M'.\end{equation}
	Analogously, we say that an ontological model of an operational theory is preparation noncontextual if
	\begin{equation} p(k|P,M) = p(k|P',M) \quad\forall M\quad\Rightarrow\quad P=P'. \end{equation}
	\end{definition}
	
	The central idea is simple: if measurements $M$ and $M'$ give the same statistics for every preparation procedure $P$, then we must say that they are in fact the same measurement, with equivalent mathematical representation, and if preparation procedures $P$ and $P'$ give the same statistics for every measurement $M$, then we must say that they are in fact the same preparation procedure, with equivalent mathematical representation.

	Note that this definition improves on Bell's definition by removing any explicit reference to quantum theory, talking about only an ``operational theory'', \ie, a theory in which we can talk about preparation procedures, measurements, and probabilities. However, this is still not the definition we're looking for. We want to be able to say whether a given probability distribution is contextual or not, as we do with the definition of nonlocality. This we shall do in the next chapter; for this one, this definition is good enough.

	We want to specialize this definition to ontological models of quantum theory, as a matter of convenience, since that's all we'll be talking about. Note that in quantum theory $p(k|P,M) = \tr\rho M_k$ is completely defined by the measurement operator $M_k$ and the quantum state $\rho$, so that's all our ontological model can take into account. More precisely

	\begin{definition}\label{def:quantumcontext}
	We say that an ontological model of quantum theory is measurement noncontextual if
	\[\xi_{k|M}(\lambda) = \xi_{M_k}(\lambda), \]
	that is, if the response function associated to the outcome $k$ of a measurement $M$ depends only on the measurement operator $M_k$. Analogously, we say that an ontological model of quantum theory is preparation noncontextual if
	\[\mu_{P}(\lambda) = \mu_\rho(\lambda),\]
	that is, if the ontic state associated to the preparation procedure $P$ depends only on the quantum state $\rho$ that is prepared.
	\end{definition}
	What else could the ontic state $\mu_{P}(\lambda)$ possibly depend on? Well, in the ontological models we discussed in sections \ref{sec:naive} and \ref{sec:bellmeu} it depended on the ``true'' basis of $\rho$, making these states preparation contextual. It could also depend on the ``true'' purification of $\rho$, or really anything that one might deem plausible or implausible. What about measurements? Well, the most famous sort of context is that of Bell's definition of contextuality: the whole PVM $M$, as do the ontological models discussed on section \ref{sec:ljbr}, but it could also be anything, such as the colour of the measurement apparatus, the latitude and longitude of the laboratory where the experiment is performed, etc.

	One final remark: if quantum theory were an ontological model of itself then definition \ref{def:spekkenscontext} (and \ref{def:quantumcontext}) would imply that it is \emph{not} contextual, since it is trivial to prove that
	\[\tr\rho M_k = \tr\rho M_k'\quad\forall\rho\quad\Rightarrow M_k=M_k'\]
	and
	\[\tr\rho M_k = \tr\sigma M_k\quad\forall M_k\quad\Rightarrow \rho=\sigma.\]
	Since it is not, the oft-heard claim that ``quantum mechanics is contextual'' is just meaningless. What one probably means with it is that any ontological model of quantum theory must be contextual, repeating a situation that happen in the area of nonlocality: quantum mechanics is obviously a local theory, in the relativistic sense, but any ontological model of quantum theory must be nonlocal, leading to the meaningless sentence ``quantum mechanics is nonlocal''.

\section{Contextuality for preparation procedures}\label{sec:preparationcontextuality}


	In this section we shall show that it is not possible to construct a preparation noncontextual	ontological model of quantum theory \cite{spekkens05}. This is not the conflict with quantum theory usually discussed, but we feel that it is appropriate to begin with it for three reasons:

	\begin{enumerate}
		\item It is independent of assumptions on determinism
		\item It is simple
		\item It is novel
	\end{enumerate}

	To begin, we'll need to prove a simple lemma about how orthogonal states are represented in the ontic space $\Lambda$. We'll see that the possibility of distinguishing orthogonal states with certainty by a single-shot measurement implies that their representations in the ontic space must have disjoint support.

	\begin{lemma}\label{lem:orthogonality}
	If two quantum states $\rho$ and $\sigma$ are orthogonal then the corresponding ontic states $\mu_\rho$ and $\mu_\sigma$ have disjoint support:
	\[\rho\sigma = 0\quad \Rightarrow\quad \mu_\rho(\lambda)\mu_\sigma(\lambda) = 0\quad\forall\lambda\]
	\end{lemma}
	\begin{proof}
	If $\rho$ and $\sigma$ are orthogonal, then they can be distinguished with certainty in a single-shot measurement. To construct one such measurement, note that the supports of $\rho$ and $\sigma$ must be orthogonal, and let $\Pi_\rho$ be the projector onto the support of $\rho$. Then
	\[\tr\rho\Pi_\rho=1\mathand\tr\sigma\Pi_\rho=0.\]
	Writing these measurements ontologically, we have
	\[\int_\Lambda\mu_\rho\xi_{\Pi_\rho}=1\mathand\int_\Lambda\mu_\sigma\xi_{\Pi_\rho}=0,\]
	so $\xi_{\Pi_\rho}(\lambda)=1$ for all $\lambda$ in the support of $\mu_\rho$, and $\xi_{\Pi_\rho}(\lambda)=0$ for all $\lambda$ in the support of $\mu_\sigma$, so the supports of $\mu_\rho$ and $\mu_\sigma$ are disjoint, and $\mu_\rho(\lambda)\mu_\sigma(\lambda)=0$ for all $\lambda$.
	\end{proof}

	We will also need the assumption that is violated by all the ontological models discussed so far:
	\begin{assumption}[Preparation noncontextuality]\label{ass:preparation}
	\[\sum_i p_i \psi_i = \sum_i q_i \phi_i\quad\Rightarrow\quad\mu_{(p_i,\psi_i)}(\lambda) = \mu_{(q_i,\phi_i)}(\lambda)\]
	\end{assumption}
	

	With the groundwork laid, we can now state the theorem and prove it.
	
	\begin{theorem}[Spekkens \cite{spekkens05}]\label{teo:spekkens}
	It is not possible to embed quantum theory into a preparation noncontextual ontological theory.
	\end{theorem}
	\begin{proof}
	Let $\phi$, $\Phi$, $\chi$, $\mathrm{X}$, $\psi$, and $\Psi$ be quantum states such that
	\begin{subequations}\label{eq:spekkens}
	\begin{align}
	\label{eq:spekkens1} 0 &= \phi\Phi = \chi\mathrm{X} = \psi\Psi \\
	\label{eq:spekkens2} \vphantom{\frac{3}{2}}\id &= \phi + \Phi = \chi + \mathrm{X} = \psi + \Psi \\
	\label{eq:spekkens3} \frac{3}{2}\id &= \phi + \chi + \psi = \Phi + \mathrm{X} + \Psi.
	\end{align}
	\end{subequations}
	That such a family of states exists can be proven by exhibiting an example in dimension 2, that can be easily embedded in higher dimensions:
	\begin{subequations}\label{eq:spekkensstates}
	\begin{align}
	\ket{\phi} &= \ket{0} & \ket{\Phi} &= \ket{1}\\
	\ket{\chi} &= \frac{1}{2}\ket{0} + \frac{\sqrt{3}}{2}\ket{1} & \ket{\mathrm{X}} &= \frac{\sqrt{3}}{2}\ket{0}-\frac{1}{2}\ket{1}\\
	\ket{\psi} &= \frac{1}{2}\ket{0} - \frac{\sqrt{3}}{2}\ket{1} & \ket{\Psi} &= \frac{\sqrt{3}}{2}\ket{0}+\frac{1}{2}\ket{1}
	\end{align}
	\end{subequations}
	A nice way to visualize the orthogonality and completeness relations \eqref{eq:spekkens} is to represent states \eqref{eq:spekkensstates} in the $\sigma_x,\sigma_z$ plane of the Bloch sphere, as done in figure \ref{fig:spekkensbloch}.

	\begin{figure}[ht]
	\centering
	\includegraphics[width=0.50\textwidth]{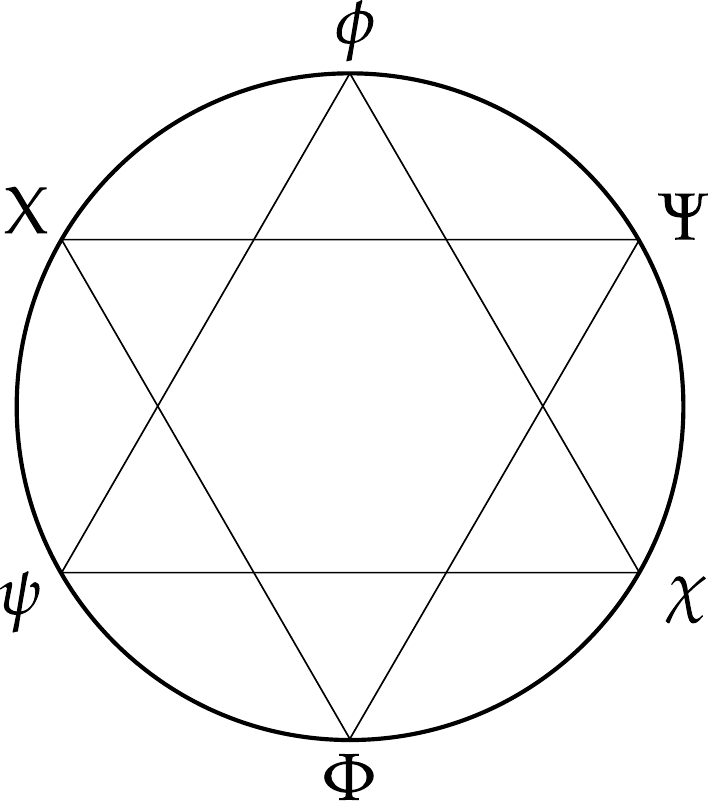}
	\caption{Representation of states \eqref{eq:spekkensstates} in the $\sigma_x,\sigma_z$ plane of the Bloch sphere. The barycenter of antipodal states or states which are connected by a triangle is $\id/2$.}
    \label{fig:spekkensbloch}
	\end{figure}

	Now we shall use lemmas \ref{lem:mixedontic} and \ref{lem:orthogonality} together with assumption \ref{ass:preparation} and relations \eqref{eq:spekkens} to derive a contradiction. 
Lemma \ref{lem:orthogonality} together with \eqref{eq:spekkens1} implies that
	\begin{equation}\label{eq:ortogonal}
	\mu_\phi(\lambda)\mu_\Phi(\lambda) = \mu_\chi(\lambda)\mu_\Chi(\lambda) =\mu_\psi(\lambda)\mu_\Psi(\lambda) = 0\quad\forall\lambda
	\end{equation} 
	Lemma \ref{lem:mixedontic}, together with assumption \ref{ass:preparation} and relations \eqref{eq:spekkens2}, implies that
	\begin{subequations}\label{eq:soma2}	
	\begin{align}
	\label{eq:somaphi}\mu_{\frac{1}{2}\id} &= \frac{1}{2}\de{\mu_\phi+\mu_\Phi}\\
			     &= \frac{1}{2}\de{\mu_\chi+\mu_\Chi}\\
			     &= \frac{1}{2}\de{\mu_\psi+\mu_\Psi},
	\end{align}
	\end{subequations}
	and together with relations \eqref{eq:spekkens3}
	\begin{subequations}\label{eq:soma3}
	\begin{align}
	\label{eq:somaminuscula}\mu_{\frac{1}{2}\id} &= \frac{1}{3}\de{\mu_\phi+\mu_\chi+\mu_\psi}\\
			     &= \frac{1}{3}\de{\mu_\Phi+\mu_\Chi+\mu_\Psi}.
	\end{align}
	\end{subequations}
	
	We shall conclude the proof by showing that the only simultaneous solution to \eqref{eq:soma2}, \eqref{eq:soma3}, and \eqref{eq:ortogonal} is the all-zero solution 
	\[\mu_\phi(\lambda)=\mu_\Phi(\lambda) = \mu_\chi(\lambda)=\mu_\Chi(\lambda) =\mu_\psi(\lambda)=\mu_\Psi(\lambda) = 0\quad\forall\lambda,\]
	which is absurd, since probability distributions can't be zero everywhere.

	The disjointness relations \eqref{eq:ortogonal} imply that for each $\lambda$ at least one of $\mu_\phi$ and $\mu_\Phi$ must be zero, and the same for the other letters. Therefore there are 8 different cases to examine, although only two are essentially different. The first one is when $\mu_\phi$, $\mu_\chi$, and $\mu_\psi$ are zero. Then \eqref{eq:soma3} implies that $\mu_\Phi$, $\mu_\Chi$, and $\mu_\Psi$ must also be zero. The second case is when $\mu_\Phi$, $\mu_\chi$, and $\mu_\psi$ are zero. Then \eqref{eq:somaphi} implies that $\mu_{\frac{1}{2}\id} = \frac{1}{2}\mu_\phi$, and \eqref{eq:somaminuscula} implies that $\mu_{\frac{1}{2}\id} = \frac{1}{3}\mu_\phi$. But the only solution to $\frac{1}{2}\mu_\phi = \frac{1}{3}\mu_\phi$ is $\mu_\phi = 0$, and we can apply the previous argument to show that all probability distributions must be zero. The six remaining cases are simply relabellings of these two.

	As the above argument applies to every $\lambda$, we have that all probability distributions are zero for every $\lambda$, and thus are not probability distributions.
	\end{proof}

\section{Gleason theorems}\label{sec:gleason}

	There are three theorems that I call ``Gleason theorems'': von Neumann's theorem \cite{vonneumann32}, Gleason's theorem \cite{gleason57} and Busch's theorem \cite{busch03}. Of these three, the most famous is certainly Gleason's\footnote{The most infamous being von Neumann's. Busch's theorem is still new.}, and that is why I chose to name this section after it. All three theorems share a similar structure: they postulate some properties that a measurement $\mu$ should have, and then prove that the only measurement that satisfies those properties is the quantum mechanical one $\mu(A) = \tr \rho A$. They can be interpreted in two ways:
	\begin{enumerate}
		\item As an axiomatic improvement, by showing that the notion of quantum state and Born's rule follow from weaker axioms.
		\item As excluding deterministic ontological theories, by saying that properties of $\mu$ should be true in $\emph{any}$ theory, not only in quantum mechanics. Then one only has to notice that Born's rule is not deterministic.
	\end{enumerate}
	If one chooses the first interpretation, all three theorems are perfectly fine, and in fact quite similar. Problems arise, however, if one insists on interpreting them as excluding deterministic ontological theories. Then von Neumann's theorem becomes \emph{foolish}\footnote{The hasty reader might wonder why learn a foolish theorem. A quick answer would be to avoid repeating mistakes of the past \cite{alicki08,zukowski09}. For a longer answer, read the section.}
 \cite{mermin93}, as its assumptions already excludes a large class of ontological theories, without good reason.

	\subsection{von Neumann's theorem}
	\begin{theorem}[von Neumann \cite{vonneumann32}]
	Let $A,B$ be self-adjoint operators, and $\mu : \selfadjoint(\Hi) \to \R$ a function such that \begin{enumerate}
		\item $\mu(\alpha A) = \alpha \mu (A)$ for real $\alpha$.
		\item $\mu(A+B) = \mu(A)+\mu(B)$ for commuting $A, B$.
		\item $\mu(A+B) = \mu(A)+\mu(B)$ for non-commuting $A, B$.
		\item $\mu(\id) = 1$
		\item $\mu(A) \ge 0$ for positive A.
	\end{enumerate}
	Then any such function can be written as \[\mu(\Pi) = \tr \rho \Pi,\] where $\rho$ is a positive operator of unit trace.
	\end{theorem}
	\begin{proof}
	Properties 1, 2, and 3 establish that $\mu$ is a linear functional on $\selfadjoint(\Hi)$, and by the Riesz lemma can be represented as an inner product $\mu(A) = \tr \rho A$. Property 4 then implies that $\rho$ has unity trace, as $\mu(\id) = \tr\rho\id = \tr\rho=1$, and property 5 implies its positivity, since in particular projectors are positive operators, and $\mu(\ketbra{\psi}{\psi}) = \tr\rho\ketbra{\psi}{\psi} = \exval{\psi}{\rho}{\psi} \ge 0 $ for all $\psi$ is the definition of positivity.
	\end{proof}

	We can see, then, that the theorem itself is quite simple, and its value resides in the strength of its assumptions, which we shall examine now. The first thing one may notice is that the theorem already makes use of the Hilbert space formalism for the observables, and the fact that the states also follow the same formalism seems almost like a tautology. But this is not the case. Quantum mechanics can already implement this formalism in experiments in a quite successful manner, and one may regard observable $A$ as just a proxy for the experiment that implements it; as $\mu$ can be any function \textit{a priori} (we don't even assume it is continuous), there is not limitation in using $\selfadjoint(\Hi)$ as its domain. We shall now proceed to examine the physical content of the assumptions.

	Assumption 1 and 2 can be interpreted as doing classical post-processing to the data of a single experiment, the measurement of a PVM $\DE{\Pi_i}$, that we define from the eigendecomposition of A. The multiplication of A by a constant is implemented just by multiplying its eigenvalues by the same constant. To implement the observable $A+B$ corresponding to the sum of commuting operators $A$ and $B$ one notices that they can be diagonalized simultaneously as $A = \sum_i a_i \Pi_i$ and $B = \sum_i b_i \Pi_i$, and so their sum $A+B = \sum_i (a_i + b_i)\Pi_i$ is just a combination and rescaling of the data coming from the $\Pi_i$ outputs. Assumptions 4 and 5 can be justified by the possibility of interpreting $\mu(\Pi_i)$ as a probability: probabilities are positive, and some outcome must happen.

	The one which is harder to justify is assumption 3, since $A$, $B$, and $A+B$ correspond to different experimental configurations: so the possibility of measuring $A+B$ just by processing the data coming from the PVMs that measure $A$ or $B$ is excluded. Its justification comes from the fact that in quantum mechanics $\tr \rho(A+B) = \tr\rho A+\tr\rho B$, and our ontological theory must reproduce its results. But this is where von Neumann slips, and to make the slip more clear, it's best to use the ontological notation, the correspondence being $\mu(A) = \xi_A(\lambda)$. So assumption 3 translates to 
		\[ \xi_{A+B}(\lambda) = \xi_{A}(\lambda)+\xi_{B}(\lambda), \]
	which is clearly overkill, since correspondence with quantum mechanics only requires that 
		\[ \int_\Lambda \mu_\rho \xi_{A+B} = \int_\Lambda \mu_\rho \xi_{A}+\int_\Lambda \mu_\rho\xi_{B}, \]
	that is, that the expected values correspond, not the values of the response functions themselves. For instance, in the Bell-Mermin model, discussed in appendix \ref{sec:bellmermin}, we can see that the response function \eqref{eq:bellmermin} is clearly linear with respect to the sum of commuting observables\footnote{Note that $A$ and $B$ commute iff $b = \alpha a$ for some real $\alpha$.} 
	\[A = a_0\id + a \cdot \sigma\mathand B = b_0\id + b\cdot \sigma = b_0\id + \alpha a\cdot \sigma,\]
	as
	\begin{align*}
	\xi_{A + B}(\psi,\lambda) &= a_0 + b_0 + \norm{a+\alpha a}\sign((a+\alpha a)\cdot(\lambda + \hat{\psi})) \\
			     &= a_0 + b_0 + \abs{1+\alpha}\norm{a}\sign(1+\alpha)\sign(a\cdot(\lambda + \hat{\psi})) \\
			     &= a_0 + \norm{a}\sign(a\cdot(\lambda+\hat{\psi})) + b_0 + \alpha\norm{a}\sign(a\cdot(\lambda+\hat{\psi})) \\
			     &= \xi_{A}(\psi,\lambda) + \xi_{B}(\psi,\lambda),
	\end{align*}
	since the values that $\xi_{A}$ assumes are the eigenvalues of $A$, and eigenvalues \emph{are} linear with respect to the sum of commuting observables. Of course, this is not true when the observables do not commute, as we can see in the following example:
	\begin{align*} 
	\xi_{\sigma_x + \sigma_z}(\psi,\lambda) &= \sqrt{2}\sign(\lambda_x+\psi_x+\lambda_z+\psi_z) \\
				  &\neq \sign(\lambda_x + \psi_x) + \sign(\lambda_z + \psi_z) \\
				  &= \xi_{\sigma_x}(\psi,\lambda) + \xi_{\sigma_z}(\psi,\lambda).
	\end{align*}

	Therefore, we must conclude that this assumption is unfounded, and if no justification can be found to it, we must abandon von Neumann's prohibition of ontological models. We shall see, however, that even if we abandon this assumption, we can still prove a von Neumann-like theorem, valid in a more restricted context: that is Gleason's theorem. More surprisingly, however, is the fact that this assumption \emph{can} be justified, by the consideration of POVMs. This realisation is what motivated the proof of Busch's theorem.

	\subsection{Gleason's theorem}

	Andrew Gleason was not concerned with von Neumann's theorem, not even with the problem of ontological models for quantum mechanics. His goal was to study the mathematical foundations of quantum mechanics, and to strengthen its axiomatic basis by showing that essentially every measure on a Hilbert space is given by Born's rule \cite{gleason57}. Its significance to the exclusion of ontological models of quantum mechanics was first noticed by Bell \cite{bell66}, who also remarked that contextual ontological models were not bound by Gleason's theorem.

	\begin{theorem}[Gleason \cite{gleason57}]
	Let $\Hi$ be a separable Hilbert space over $\C$ with $\dim\Hi \ge 3$, and $\mu : \pro(\Hi) \to [0,1]$ a function such that $\sum_i \mu(\Pi_i) = 1$ for any PVM $\DE{\Pi_i}$. Then any such function can be written as \[\mu(\Pi_i) = \tr \rho \Pi_i,\] where $\rho$ is a positive operator of unity trace.
	\end{theorem}

	The proof of this theorem is already well-known, and a bit boring, so we shall omit it. The interested reader may find it in the original work \cite{gleason57}, or in the clearer version by Bell \cite{bell66}.

	It is easy to see that von Neumann's $\mu$ functions satisfy all the properties of Gleason's $\mu$ functions, and continue to do so even if we drop his questionable assumption 3, so it is certainly possible to interpret Gleason's theorem as a ``reasonable'' von Neumann theorem, with weaker assumptions. Also notice that Gleason's assumptions are explicitly non-contextual, by assuming that $\mu(\Pi_i)$ is only a function of the projector $\Pi_i$, and not of the whole PVM.

	\subsection{Busch's theorem}

	Paul Busch was concerned with the justification of von Neumann's assumption 3. He noticed that if one measures a POVM $\DE{E_i}$ instead of a PVM, then it is possible to have in a single experiment two outcomes $E_0$ and $E_1$ that do not commute\footnote{In fact, this happens in all non-trivial POVMs..}, so it is perfectly natural to demand that $\mu\de{E_0 + E_1} = \mu\de{E_0} + \mu\de{E_1}$, since one can measure $E_0+E_1$ just by combining the outcomes corresponding to $E_0$ and $E_1$. He then restricted assumption 3 to sums of effects belonging to a single POVM, and was able to derive Born's rule from it, thus resurrecting von Neumann's theorem \cite{busch99}. Later he realized that the form of his theorem was actually closer to Gleason's than von Neumann's; to obtain it from Gleason's one only has to demand $\sum_i \mu(\Pi_i) = 1$ to be true for POVMs, instead of just form PVMs. Interpreted in this way, his theorem is a much stronger version of Gleason's with a much simpler proof \cite{busch03}.

	The proof presented here mostly follows the one presented in \cite{fuchs02}, with the difference that it does not require the domain of $\mu$ to be extended.

	\begin{theorem}[Busch \cite{busch03}]
	Let $\Hi$ be a separable Hilbert space over\footnote{$\Q[i]$ is the field extension of the rationals $\Q$ with the imaginary number $i$, $\Q[i] = \DE{a+ib:a,b\in\Q}$.} $\Q[i]$ or $\C$, and $\mu : \eff(\Hi) \to [0,1]$ a function such that $\sum_i \mu(E_i) = 1$ for any POVM $\DE{E_i}$. Then any such function can be written as \[\mu(E_i) = \tr \rho E_i,\] where $\rho$ is a positive operator of unity trace.
	\end{theorem}
	\begin{proof}
	The proof begins by noticing that $\mu$ is in fact a linear functional on $\eff(\Hi)$. From that, the Riesz lemma establishes that it can represented as an inner product. Positivity and normalization of $\rho$ then comes from the positivity and normalization of $\mu$. We shall first prove the case where $\Hi$ is over the complex rationals, and later extend the proof to the continuum.

	First note that if $E$ is an effect, $\id -E$ is also an effect. Then considering the POVMs $\DE{E,\id-E}$ and $\DE{E_1,E_2,\ldots,E_n,\id-E}$, where $\sum_i E_i = E$, we see that $\mu(E) = \sum_i \mu(E_i)$. Considering the particular case $E_i = E/n$, we get that $\mu(E) = n\mu(E/n)$. On the other hand, if we consider $E = mF$ and $E_i = F$, we get $\mu(mF) = m\mu(F)$. Combining these two cases, we see that $\mu(\frac{m}{n}E) = m\mu(\frac{1}{n}E) = \frac{m}{n}\mu(E)$, that is, $\mu(qE) = q\mu(E)$ for $q \in \Q^+$ whenever both $qE$ and $E$ are effects. Wrapping up, we have that 
	\[\mu(E) = \sum_i q_i \mu(E_i)\]
	for rational $q_i$ whenever $q_i E_i$ are effects, so $\mu$ already has some restricted linearity. If we can remove the restriction that $q_i E_i$ are effects, we get full linearity on $\eff(\Hi)$, and that's what we'll do now.
	
	Consider the effects $E$ and $F \le E$. Then $E = E-F + F$, and $\mu(E) = \mu(E-F) + \mu(F)$, so $\mu(E-F) = \mu(E) - \mu(F)$. Consider now $E,F,G \in \eff(\Hi)$ and $p,q \in \Q^+$ such that $E = pF-qG$, but at least one of $p$ and $q$ is larger than unity, so $pF$ and $qG$ are not necessarily effects. Without loss of generality, let $p \ge q$. Then $\frac{1}{p}E$, $F$, and $\frac{q}{p}G$ are all effects, and by the property we just proved, $\mu(\frac{1}{q}E) = \mu(F) - \mu(\frac{q}{p}G)$, so $\mu(E) = p\mu(F)-q\mu(G)$ and
	\[\mu(E) = \sum_i q_i \mu(E_i)\]
	for \emph{any} rational $q_i$, so we have full linearity on $\eff(\Hi)$. Let then $\DE{E_i}_{i=1}^{d^2}$ be a MIC-POVM and, as such, a basis for $\Hi$. Then any effect $E$ can be written as $E = \sum_{i=1}^{d^2} q_i E_i$ for $q_i \in \Q$ (a moment's thought will convince you that complex numbers aren't allowed). We can now define $\rho$ by solving the $d^2$ equations $\tr \rho E_i = \mu(E_i)$, and see that
	\[\mu(E) = \sum_{i=1}^{d^2} q_i \mu(E_i) = \sum_{i=1}^{d^2} q_i \tr \rho E_i = \tr\de{\rho\sum_{i=1}^{d^2} q_i E_i} = \tr \rho E.\]
	Positivity of $\rho$ comes from considering the case where $E$ is a one-dimensional projector:
	\[ 0 \le \tr \rho E = \tr \rho \ketbra{\psi}{\psi} = \exval{\psi}{\rho}{\psi}.\]
	The unity of the trace comes from 
	\[ 1 = \sum_i \mu(E_i) = \sum_i \tr \rho E_i = \tr\de{\rho\sum_i E_i} = \tr \rho.\]
	This completes the proof for $\Q[i]$. To extend it to the continuum, note again that if $E\ge F$, then $\mu(E) = \mu(E-F) + \mu(F)$, and so $\mu(E) \ge \mu(F)$. Let then $p_i$ and $q_i$ be sequences of rational numbers tending to the real number $\alpha$ such that $p_i \le \alpha \le q_i$. We have $p_i E \le \alpha E \le q_i E$, and as such $p_i\mu(E)\le\mu(\alpha E)\le q_i\mu(E)$, so $\mu(\alpha E) = \alpha \mu(E)$. From this fact, one can now retrace the proof and see that it also holds for $\C$.
	\end{proof}

	The reason that we decided to highlight the fact that Busch's theorem holds for $\Q[i]$ is that the original Gleason theorem fails for it, hinting that traditional contextuality might have problems dealing with subsets of $\C$ \cite{meyer99,pitowsky83}. This feature of Busch's theorem was first noticed in \cite{caves04}.

\subsection{Wrapping up}

	Busch's theorem is clearly superior to von Neumann's in every way, but this is not true for Gleason's: they can be interpreted in different ways. Busch's shows that there can't be a non-contextual model capable of reproducing quantum mechanics in any dimension, while Gleason's opens up the possibility of such a model existing in dimension two, if we only care about projective measurements. That such a model exists can be seen by looking at the Bell-Mermin model in appendix \ref{sec:bellmermin}; but if, like Gleason, the reader is not interested in the question of ontological theories, but in which measures are allowed given the Hilbert space structure of observables, the following counterexample\footnote{Due to Marcelo Terra Cunha and Rafael Rabelo.} should suffice:
	\[\mu_\psi(\phi) = \frac{1}{2}(1+\cos (n \cos^{-1}(\hat{\phi}\cdot\hat{\psi}))),\quad \text{odd }n\]
	Note that for $n=1$ this formula is simply Born's rule. 

	It is easy to check that 
	\begin{align*}
	\sum_i \mu(\Pi_i) &= 
\mu_\psi(\phi) + \mu_\psi(\id -\phi) \\
	  &= \frac{1}{2}(1+\cos (n \cos^{-1}(\hat{\phi}\cdot\hat{\psi}))) + \frac{1}{2}(1+\cos (n \cos^{-1}(\hat{\phi}\cdot\hat{\psi})+\pi)) \\
	  &= \frac{1}{2}(1+\cos (n \cos^{-1}(\hat{\phi}\cdot\hat{\psi}))) + \frac{1}{2}(1-\cos (n \cos^{-1}(\hat{\phi}\cdot\hat{\psi}))) \\
	  &= 1,
	\end{align*}
	as required in Gleason's assumptions. 

	To see that for $n\ge 3$ this formula can't equal Born's rule, notice that
	\[\tr\psi\phi = \frac{1}{2}(1+ \hat{\phi}\cdot\hat{\psi}) \]
	only has one root, if considered as a function of the angle $\cos^{-1}(\hat\phi\cdot\hat\psi)$, whereas our $\mu_\psi(\phi)$ has $n$ roots.

\section{The Kochen-Specker theorem}\label{sec:kochenspecker}	

	A corollary of the Gleason theorem is that one can't embed quantum theory in a noncontextual ontological model if $\dim\Hi\ge 3$, since the Born rule is explicitly noncontextual and non-deterministic; a direct proof of this fact might seem superfluous. But one might not like its assumptions: after all, it already assumes a fair bit of structure that is not quite needed and, more importantly, it needs to assume that the quantum valuation $\mu(\Pi_i)$ is defined for a continuous amount of projectors, which of course can never have experimental justification. This was the motivation\footnote{The motivation can come from Gleason's theorem, or from a 1960 work of Specker \cite{specker60,seevinck11}, that was independent of Gleason and also contained a ``continuous'' proof of contextuality.} for Simon Kochen and Ernst Specker to develop a \emph{finite} proof of noncontextuality, finding an inconsistency in any deterministic assignment of values to a set of experiments realizable in quantum mechanics \cite{kochen67}. Another motivation to present it here is that it proves the claim in section \ref{sec:bellmeu} that noncontextual deterministic ontological models can not describe two-outcome PVMs.

	In modern parlance, the Kochen-Specker theorem is referred to as a proof of state-independent contextuality, as the logical contradiction found depends only on the structure of quantum observables, and not on the statistics from the measurement of specific states. This situation contrasts, of course, with proofs of state-dependent contextuality, which we shall explore mainly on the next chapter.

	More specifically, their proof says that we can't attribute deterministic values $\xi_{\Pi_i}(\lambda)$ to a set of projectors $\{\Pi_i\}_{i=1}^{117}$ in dimension three respecting the quantum mechanical observation that in the measurement of a PVM one answer (and only one answer) always occurs. An elegant way to proceed with the proof is to represent this set of projectors in an orthogonality graph (where each vertex corresponds to a projector, and two vertices are connected iff the corresponding projectors are orthogonal), and map the quantum mechanical observation into two rules for colouring the graph:
	\begin{enumerate}
		\item Two connected vertices can't both have the value $1$ -- If two projectors $\Pi_i$ and $\Pi_j$ are orthogonal, they can be measured simultaneously, and therefore $\xi_{\Pi_i}(\lambda)$ and $\xi_{\Pi_j}(\lambda)$ can't both equal $1$.
		\item In a loop of three connected vertices, one of them must have the value $1$ -- If three projectors are mutually orthogonal, they form a PVM, and in a PVM one answer (and only one answer) always occurs.
	\end{enumerate}

	The proof concludes by showing that no such colouring of the graph can exist, and therefore one can't attribute deterministic values to this set of projectors. We shall, however, omit it. Even though it is quite beautiful, the proof is mainly of historical interest, as simpler proofs have hitherto been found. We refer the interested reader to the original paper, or the excellent exposition of it by Cabello \cite{cabello96}.

\subsection{An 18-projector proof by Cabello, Estebaranz, and García-Alcaine}

	The simplest (with fewest projectors) such no-colouring proof that we currently know\footnote{We do know that in dimensions 3 and 4 there are no no-colouring proofs with 17 projectors or less \cite{cabello06,arends11}.} was found in 1996 by Cabello, Estebaranz, and García-Alcaine \cite{cabello96b}. In contrast with Kochen-Specker's 117 projectors, it needs only 18 to generate a contradiction. These projectors are represented in figure \ref{fig:adan18}, where $v = (a,b,c,d)$ is just a shorthand notation for the projector onto $\ket{v} = a\ket{0} + b\ket{1} + c\ket{2} + d\ket{3}$. This figure does not represent an orthogonality graph, which would be quite cumbersome, but an orthogonality hypergraph, where sets of four commuting projectors are connected by edges of the same colour.

	\begin{figure}[ht]
	\centering
	\includegraphics[width=1.00\textwidth]{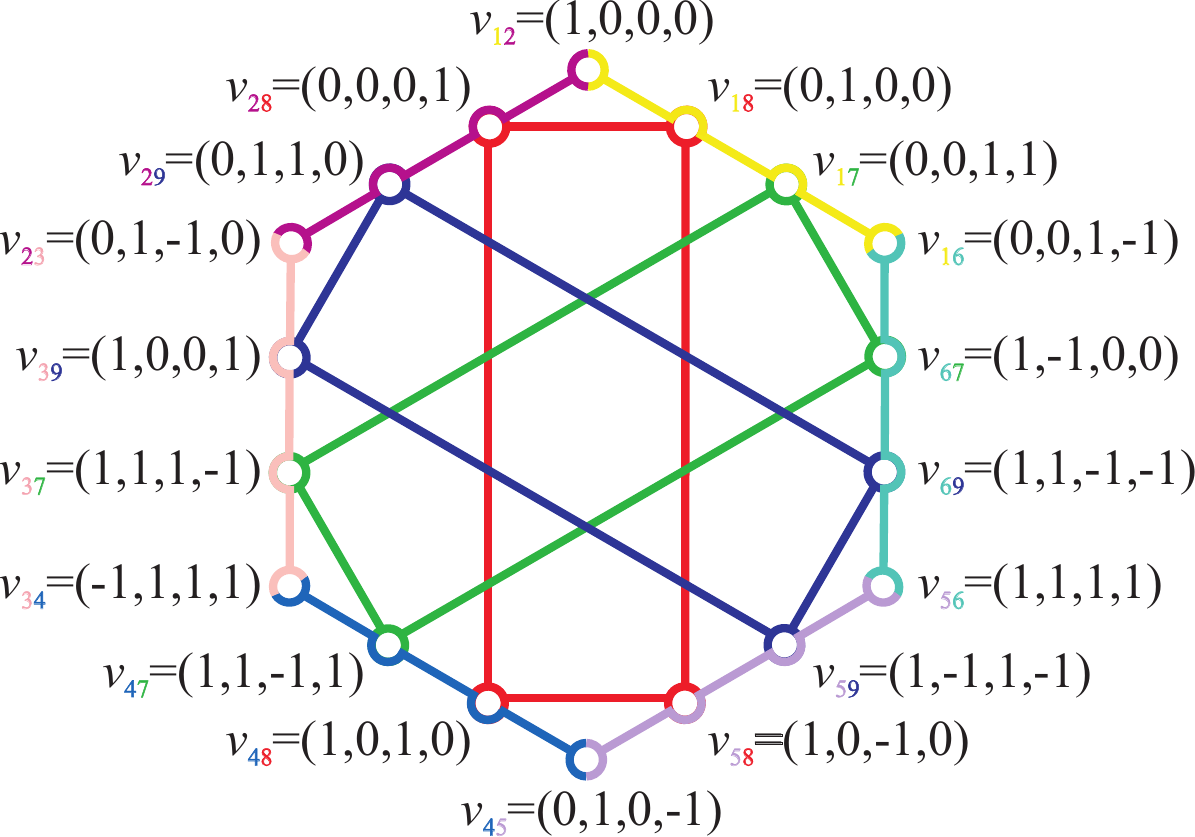}
	\caption{Vectors for the 18-projector proof of the Kochen-Specker theorem. Reproduced from \cite{cabello08} with permission from the author.}
	\label{fig:adan18}
	\end{figure}

	One could in fact proceed to prove directly that it is non-colourable (there are few non-equivalent potential colourings), but it is more elegant to use a parity argument: we know that in each context we must have one answer $1$, so the sum over all answers in all contexts must be $9$. But if we do this sum projector by projector, we see that each projector appears in exactly two contexts, and likewise each answer appears twice, so the sum over them must be an even number, a contradiction.

\subsection{A 13-projector proof by Yu and Oh}\label{sec:yuoh}

	Shockingly, more recently it has been found that a non-colourable graph is \emph{not} necessary to prove state-independent contextuality. Yu and Oh \cite{yu12} have found such a proof in dimension 3 based on a set of 13 projectors that \emph{does} have a colouring that obeys rules 1 and 2. They argue that every possible colouring of their graph contradicts another prediction of quantum theory. The orthogonality graph is represented in figure \ref{fig:yugioh}, and its quantum realization is given by the vectors 
	\begin{align*}
	z_1 &= (1,0,0)&  h_0 &= (1,1,1)&  y_1^+ &= (0,1,1)\\
	z_2 &= (0,1,0)&  h_1 &= (-1,1,1)& y_1^- &= (0,1,-1)\\
	z_3 &= (0,0,1)&  h_2 &= (1,-1,1)& y_2^+ &= (1,0,1)\\
	&&		 h_3 &= (1,1,-1)& y_2^- &= (-1,0,1)\\
	&&		&&		  y_3^+ &= (1,1,0)\\
	&&		&&		  y_3^- &= (1,-1,0)
	\end{align*}
	where $r = (a,b,c)$ is just a shorthand notation for the projector onto $\ket{r} = a\ket{0} + b\ket{1} + c\ket{2}$. It is important for the proof that this is actually the unique quantum realization of the orthogonality graph up to a global unitary transformation, which is trivial to prove.

\begin{figure}[ht]
	\centering
	\includegraphics[width=1.00\textwidth]{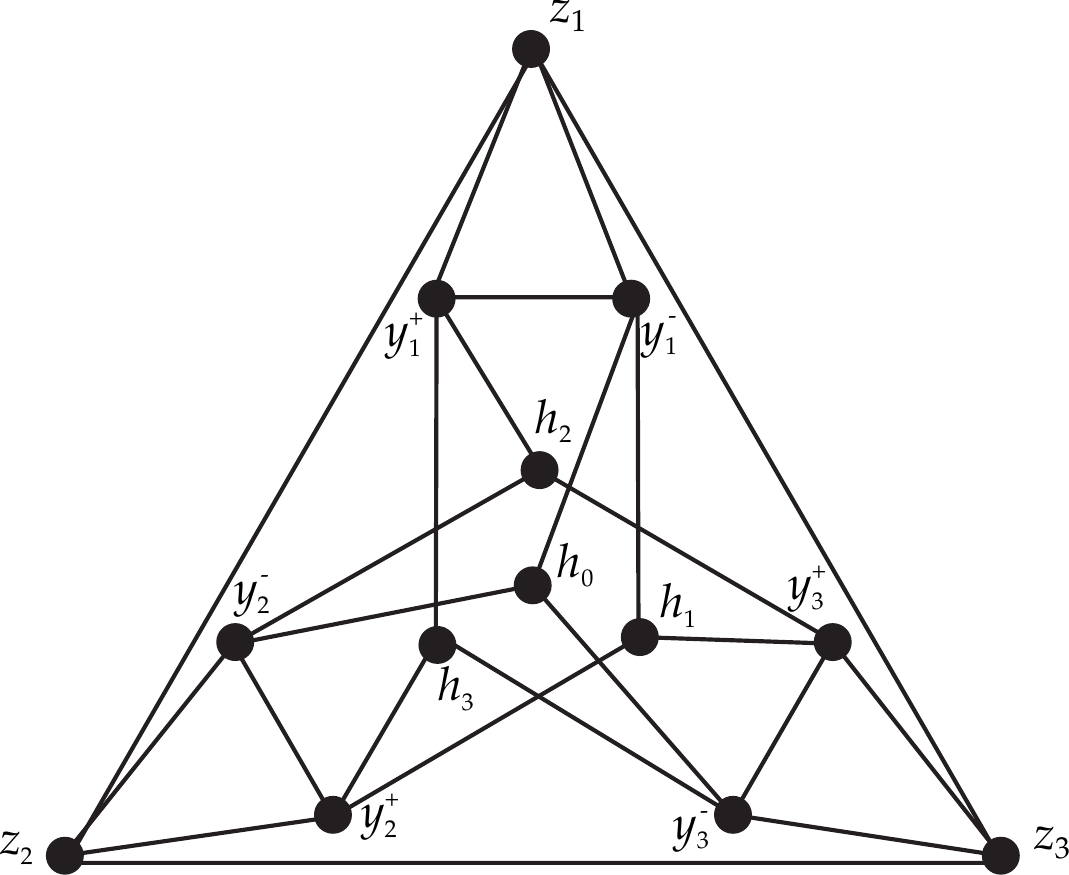}
	\caption{Orthogonality graph for the proof of Yu and Oh. Reproduced from \cite{cabello12} with permission from the authors.}
	\label{fig:yugioh}
\end{figure}

	To obtain the contradiction with quantum mechanics, first note that no two $h_i$ can be assigned $1$ simultaneously. We shall prove this by contradiction. By the symmetry of the graph, there are only two cases:
	\begin{enumerate}
		\item Assume that $\xi_{h_0}(\lambda) = \xi_{h_1}(\lambda) = 1$. Then by the KS rules we must assign $0$ to $y_2^\pm$ and $y_3^\pm$, which oblige us to assign $1$ to $z_2$ and $z_3$, a contradiction.
		\item Assume that $\xi_{h_1}(\lambda) = \xi_{h_2}(\lambda) = 1$. Then by the KS rules we must assign $0$ to $y_1^\pm$ and $y_2^\pm$, which oblige us to assign $1$ to $z_1$ and $z_2$, a contradiction.
	\end{enumerate}
	This implies that $\sum_i \xi_{h_i}(\lambda) \le 1$, and furthermore that \[ \sum_i \int_\Lambda \mu_\psi \xi_{h_i} \le 1.\] But the \textit{lhs} must be equal to the quantum expectation value $\sum_i \tr \psi h_i$; since \[ \sum_i h_i = \frac{4}{3} \id, \] we get that $\sum_i \tr \psi h_i = 4/3$ for any state, a contradiction.

\subsection{A 9-observable proof by Peres and Mermin}

	Last but not least, we'd like to present the beautiful proof of the Kochen-Specker theorem done in 1990 by Asher Peres and David Mermin \cite{peres90,mermin90}, the Peres-Mermin square. It uses 9 four-dimensional observables, so in some sense it is larger than the previous two proofs, and also older; but it is also quite elegant, and so it might seem smaller to the human mind.

	Let 
	\begin{equation}\label{eq:peresmerminsquare}
	A = \begin{pmatrix}
	\sigma_z \otimes \id & \id \otimes \sigma_z & \sigma_z \otimes \sigma_z \\
	\id \otimes \sigma_x & \sigma_x \otimes \id & \sigma_x \otimes \sigma_x \\
	\sigma_z \otimes \sigma_x & \sigma_x \otimes \sigma_z & \sigma_y \otimes \sigma_y 
	\end{pmatrix}
	\end{equation}
	be the Peres-Mermin square, where $\sigma_x$, $\sigma_y$, and $\sigma_z$ are Pauli matrices. Note that observables $A_{ij}$ that lie in the same line or column always commute, so they are simultaneously measurable, and we should be justified in assigning them a predefined value $\mu(A_{ij}) = \xi_{A_{ij}}(\lambda) \in \{-1,+1\}$. But also note that the product of the observables in each line or column is always plus or minus identity, relation that our predefined values should also respect. More specifically, this reasoning leads us to the relations
	\begin{align*}
	\mu(\sigma_z \otimes \id)\mu(\id \otimes \sigma_z)\mu(\sigma_z \otimes \sigma_z) &= +1\\
	\mu(\id \otimes \sigma_x)\mu(\sigma_x \otimes \id)\mu(\sigma_x \otimes \sigma_x) &= +1\\
	\mu(\sigma_z \otimes \sigma_x)\mu(\sigma_x \otimes \sigma_z)\mu(\sigma_y \otimes \sigma_y) &= +1\\
	\mu(\sigma_z \otimes \id)\mu(\id \otimes \sigma_x)\mu(\sigma_z \otimes \sigma_x) &= +1\\
	\mu(\id \otimes \sigma_z)\mu(\sigma_x \otimes \id)\mu(\sigma_x \otimes \sigma_z) &= +1\\
	\mu(\sigma_z \otimes \sigma_z)\mu(\sigma_x \otimes \sigma_x)\mu(\sigma_y \otimes \sigma_y) &= -1
	\end{align*}
	Note now that each predefined value appears twice in the \textit{lhs}, so the product over all of them must be $+1$. But the product over the \textit{rhs} is $-1$, a contradiction.
	
\section{Ontological excess baggage}\label{sec:excessbaggage}

	What motivated Bell to prove his famous theorem was his observation that the ontological theory of de Broglie-Bohm \cite{bohm52} has a grossly nonlocal character \cite{bell66}. A natural question for him was, then, whether this nonlocality was particular of Bohm's mechanics or actually a general character of any ontological theory \cite{bell64}. 

	In that same paper, however, Bell also noticed that to study a spin system within Bohm's theory he had to include the position degree of freedom, and reduce spin measurements to position measurements. But by doing so he enlarged the number of real parameters required to describe a single qubit from two to countable infinity, and worse, the number of ontological states had to be uncountable infinity. 

	Hardy then asked whether this is a general feature of ontological theories, or just a particularity of Bell's model for a spin in Bohm's theory, and found that the answer is yes \cite{hardy04}, naming this feature \textit{ontological excess baggage}. His theorem is the subject of this section. 

	A perhaps more simple (certainly more direct) illustration of the ontological excess baggage theorem can be found in the naïve ontological theory described in section \ref{sec:naive}, where we identify the ontic space $\Lambda$ with the space of pure states $\pro\Hi$, thus forcing $\Lambda$ to have the same cardinality as it, that is, uncountable infinity.

	\begin{theorem}[Hardy \cite{hardy04}]
	In any ontological embedding of quantum theory the ontic space $\Lambda$ is infinite.
	\end{theorem}
	\begin{proof}
	Let $\psi$ and $\phi$ be two pure quantum states, $\psi \neq \phi$. Then there is a measurement $\psi$ for which $\tr (\psi\psi) = 1$, whereas $\tr (\psi\phi) < 1$. Writing these measurements ontologically, we have \[\int_\Lambda \mu_\psi \xi_\psi = 1\mathand\int_\Lambda \mu_\phi \xi_\psi < 1,\]
	that is, $\xi_\psi(\lambda) = 1$ for all $\lambda$ in the support of $\mu_\psi(\lambda)$, but there is a $\lambda_0$ in the support of $\mu_\phi(\lambda)$ for which $\xi_\psi(\lambda_0)<1$. Consequently, $\lambda_0$ is not in the support of $\mu_\psi(\lambda)$, and we see that different ontic states must have different supports. This constitutes an injection of $\pro\Hi$ into $\pro(\Lambda)$, \ie, the set of distinct subsets of $\Lambda$, thus proving that $\pro(\Lambda)$ is uncountable. This is only possible if $\Lambda$ itself is infinite (though not necessarily uncountable).
	\end{proof}
	This proof is based on the one presented in\footnote{Note that Spekkens' claim that $\Lambda$ itself is uncountable is incorrect.} \cite{spekkens07}. 

	One might wonder whether this argument can be extended to show that $\Lambda$ must be uncountable; after all, in all our examples it is, and we have not considered all the information we have: notice that it is never true that the support of $\mu_\phi$ contains the support of $\mu_\psi$ -- they are pairwise incomparable -- so we have an injection into a subset of $\pro(\Lambda)$, which might have a smaller cardinality than it. But this hope is unfounded: there is a set $Z$ of subsets of $\N$ that has pairwise incomparable members but continuous cardinality. This was proved by Martin Goldstern as an answer to a MathOverflow question by the author \cite{mogoldstern12}.

	\begin{theorem}
	There is a set $Z$ of subsets of $\N$ that has pairwise incomparable members but continuous cardinality.
	\end{theorem}
	\begin{proof}
	For any subset $A \subseteq \N$, let $X_A = \{2n:n\in A\}$, $Y_A = \{2n+1:n\not\in A\}$, and $Z_A =X_A \cup Y_A$. Then the set of all $Z_A$ is uncountable, since there is an injection of $\pro(\N)$ into it. Also note $X_A \cap Y_A = \emptyset$, and therefore $Z_A \subseteq Z_B$ implies that $X_A \subseteq X_B$ and $Y_A \subseteq Y_B$. This in turn implies that $A \subseteq B$ and $B \subseteq A$, hence $A=B$. So the $Z_A$ are pairwise incomparable.
	\end{proof}

	But why is Hardy's theorem interesting? After all, if we're not bothered by the fact that the set of quantum states $\state(\Hi)$ is uncountable, why should we be bothered by the fact that $\Lambda$ is infinite? It all has to do with the status of the pure states. If they're not ontological, the description of the Bloch ball as a vector space of dimension three is perfectly natural. But if we insist in giving ontological status to $\ket{0}$ and $\ket{+}$, it becomes a mystery the identification of the preparation procedures $\DE{\de{\frac{1}{2},\ket{0}},\de{\frac{1}{2},\ket{1}}}$ and $\DE{\de{\frac{1}{2},\ket{+}},\de{\frac{1}{2},\ket{-}}}$ or, ontologically speaking, the states $\frac{1}{2}\mu_{\ket{0}}+\frac{1}{2}\mu_{\ket{1}}$ and $\frac{1}{2}\mu_{\ket{+}}+\frac{1}{2}\mu_{\ket{-}}$. In fact, if we remember theorem \ref{teo:spekkens}, we know that we \emph{can't} do this identification, as it is precisely the assumption of preparation noncontextuality, which we showed to be untenable. But if we don't do this identification, the Bloch ball must explode: the set of all ontic states must be the infinite-dimensional set of probability distributions over the pure states.

\section{How to make an ontological theory?}

	In light of all these no-go theorems (and ``please-don't-go'' theorems), one is left to wonder how ugly it would look deterministic ontological models that reproduced all of quantum theory (as opposed to the restricted models presented in sections \ref{sec:bellmeu}, \ref{sec:kochenspeckermodel}, and appendix \ref{sec:bellmermin}). In fact, they don't look so bad on the paper, as their necessary ugliness is more philosophical than mathematical. There are, of course, models that are quite intricate, such as de Broglie-Bohm's theory. We shall ignore it, however, as we feel that an appropriate exposition of it would be too much of a digression. What we shall present is the contextual model proposed by Bell in his critique of the Gleason theorem \cite{bell66}, together with a $\psi$-epistemic modification of it \cite{lewis12}.

	\subsection{The Bell model}

	This $\psi$-ontic model was proposed by Bell in \cite{bell66}; we present it here as rendered in \cite{lewis12}.

	The ontic space for this model is $\Lambda = \pro\Hi \times [0,1]$, the ontic state is \[\mu_\psi(\lambda_\psi) = \delta(\lambda_\psi-\psi),\] and the response functions are given by 
	\[ \xi_{k|\phi}(\lambda_\psi,\lambda) = \De{\sum_{i=0}^{k-1}\tr\lambda_\psi\phi_i < \lambda \le \sum_{i=0}^k\tr\lambda_\psi\phi_i},\]
	where $[\phantom{m}]$ are Iverson brackets\footnote{Defined as $[P]=1$ if the proposition $P$ is true and $[P]=0$ otherwise.}, the empty sum $\sum_{i=0}^{-1}\tr\lambda_\psi\phi_i$ is $0$, and normalization requires us to set $\xi_{0|\phi}(\lambda_\psi,0)~=~1$. Note that for $\dim\pro\Hi = 2$ this model reduces to the one discussed is section~\ref{sec:bellmeu}.

	This model is easily seen to be contextual, since $\xi_{k|\phi}(\lambda_\psi,\lambda)$ depends non-trivially on the whole PVM $\phi$.

	\subsection{The Lewis-Jennings-Barrett-Rudolph model}\label{sec:ljbr}

	This model \cite{lewis12} was proposed as a complement to the \textsc{pbr} theorem (theorem \ref{teo:pbr}), showing that it is in fact possible to make a contextual $\psi$-epistemic model that reproduces quantum mechanics. With it, we complete the discussion of $\psi$-ontic and $\psi$-epistemic models that began in section \ref{sec:psiontic}.

	As this model is a bit complicated, we shall study first its version for dimension 2, in order to clarify the ideas, and then proceed to the general case. The response functions used are the same ones as the previous model, whereas the ontic states will be modified in order to become $\psi$-epistemic.

	Let $\hat z$ correspond to the north pole of the Bloch sphere, and $u \cdot \hat z = \cos\theta_u$ define the polar angle $\theta_u$ of a unit vector $u$. Then we can define the northern hemisphere $\north$ as the set of vectors with $\theta_u < \pi/2$, and label the measurement $\DE{\phi_0,\phi_1}$ in such a way that $\theta_{\phi_0} \le \theta_{\phi_1}$.

	The model is based on the observation that if $\hat{\psi} \in \north$, then the probability $\tr\psi\phi_0$ will be strictly larger than $0$ for any measurement $\phi_0$. This observation has two consequences. The first is that we can define a lower bound $f(\psi)$ for $\tr\psi\phi_0$ that does not depend on $\phi_0$, as \[\tr\psi\phi_0 = \frac{1}{2}(1+\cos\theta_{\psi\phi_0}) \ge \frac{1}{2}(1+\cos(\theta_\psi+\pi/2)) = f(\psi).\]

	The second is that there exists a set of ontic states 
	\[\Lambda_\north = \DE{ (\lambda_\psi,\lambda) : \lambda_\psi \in \north\mathand 0 \le \lambda < f(\psi) } \]
	such that for any state $(\lambda_\psi,\psi)$ in it we have that $\xi_{\phi_0}(\lambda_\psi,\lambda)=1$. Using all this, we can define a ontic state $\mu_\psi$ for $\hat\psi\in\north$:
	\[\mu_\psi(\lambda_\psi,\lambda) = \delta(\lambda_\psi-\psi)\Theta\de{\lambda-f(\psi)} + f(\psi) U_{\Lambda_\north}, \]
	where $U_{\Lambda_\north}$ is the uniform distribution on $\Lambda_\north$. Notice that all these states overlap in the set $\Lambda_\north$. The quantum statistics are recovered by
	\begin{align*}
	p(0|\psi,\phi) &= \int_\Lambda \mu_\psi \xi_{0|\phi} \\
			&= \int_\Lambda \Theta\de{\lambda-f(\psi)}\Theta(\tr\psi\phi_0-\lambda) + f(\psi) \int_\Lambda U_{\Lambda_\north} \Theta(\tr\lambda_\psi\phi_0-\lambda) \\
			&= \tr\psi\phi_0 -f(\psi) + f(\psi)\int_\Lambda U_{\Lambda_\north} \\
			&= \tr\psi\phi_0.
	\end{align*}
	For the case $\hat\psi\not\in\north$, we let $\mu_\psi(\lambda_\psi,\lambda) = \delta(\lambda_\psi-\psi)$, as usual.

	To make the generalization to dimension $d$, label the measurement $\DE{\phi_i}$ in such a way that $\tr\Pi\phi_0 \ge \tr\Pi\phi_1 \ge \ldots \ge \tr\Pi\phi_{d-1}$, where $\Pi$ is an arbitrary state. Now we want to define the analogue of $\north$, \ie, a set $\north'$ such that for any $\psi$ in it we have $\tr\psi\phi_0 > 0$. To do that, first note that $\tr\Pi\phi_0 \ge 1/d$, since $\tr\Pi\phi_i$ are the elements of a probability vector. Now note that $\tr\psi\phi_0 = 0$ implies that $\psi \le \id - \phi_0$, so $\tr\Pi\psi \le \tr\Pi(\id-\phi_0) \le 1 - 1/d$, and therefore $\tr\Pi\psi > (d-1)/d$ implies that $\tr\psi\phi_0>0$. With that in hand, we can now proceed to finding the analogue of $f(\psi)$, \ie, a lower bound on $\tr\psi\phi_0$ that does not depend on $\phi_0$. Since its existence is clear, we shall not bother looking for an explicit expression and just call it $f'(\psi)$. The analogue of $\Lambda_\north$ is then
	\[ \Lambda'_\north = \DE{ (\lambda_\psi,\lambda) : \lambda_\psi \in \north' \mathand 0 \le \lambda < f'(\lambda_\psi) },\]
	and the ontic state, for $\psi \in \north'$, is
	\[ \mu_\psi = \delta(\lambda_\psi-\psi)\Theta(\lambda-f'(\psi)) + f'(\psi) U_{\Lambda'_\north}. \]
	For $\psi\not\in\north'$, we let $\mu_\psi(\lambda_\psi,\lambda) = \delta(\lambda_\psi-\psi)$, as usual. This makes the model not ``maximally $\psi$-epistemic'', that is, it is not true that for every pair of non-orthogonal states $\psi$ and $\phi$ the ontic states $\mu_\psi$ and $\mu_\phi$ have a non-zero overlap. This raises the question: is a ``maximally $\psi$-epistemic'' model possible? This question was raised by the authors of \cite{lewis12} themselves, and answered by George Lowther and Scott Aaronson in the affirmative \cite{moaaronson12}.

%
%

\chapter{Revealing surrealism}\label{cha:prob}

\epigraph{Make it simple, because I can only understand simple things.}{Asher Peres}

	Reading the previous chapter must have felt like walking in sand, with the definitions and assumptions being challenged and changed all the time. This is unfortunate, but necessary for such a discussion of the foundations of quantum mechanics. In this chapter, however, we shall use what we learned and develop a final definition of contextuality, which will serve as a solid foundation for the work ahead. 

	Instead of trying to find an ontological embedding of quantum theory, we shall just accept that it can't be done, and try to characterize exactly which parts of quantum mechanics \emph{can't} be embedded in an (noncontextual) ontological theory. We shall do this by examining the probability distributions over certain events\footnote{``Which events?'', you ask. That is the question; for a partial answer, read the rest of the chapter.}: if such a probability distribution can't be reproduced by a noncontextual ontological theory, we shall deem it truly quantum. What for, you ask? These probability distributions will be a resource to do what is impossible in classical theories: quantum computation with an exponential speedup and quantum distribution of cryptographic keys, among other things. In other words, \emph{quantum magic}.

	\section{The correct definition of contextuality}
	
	The first thing we need to do is to obtain our final definition of contextuality. As we discussed in section \ref{sec:spekkenscontextuality}, we need a definition that is not specifically about quantum mechanics, but instead about probability distributions, as is the case of the definition of locality. This need was recognized by Robert Spekkens in 2005 \cite{spekkens05}, but he stopped short of doing that: Spekkens arrived at a definition that talked about ontological models instead of quantum theory. His definition (at least, the part of it about measurements) can easily be turned into a definition that only talks about probability distributions, it shall be our definition \ref{def:contextualitywrong}. However, we shall argue that this definition misses the essential point about contextuality.

	This necessity was also recognized by Adam Brandenburger and Noson Yanofsky in 2008 \cite{brandenburger08}, but this work limited itself to translating the various notions of contextuality that exist in the literature into statements about probability distributions. It did not try to judge them and obtain a final definition of contextuality.

	Prompted by the discovery\footnote{Or rather its publication in Physical Review Letters.} of the Klyachko inequality\footnote{Note that these papers claim to exclude \emph{any} ontological models, including contextual ones. This claim is incorrect.} \cite{klyachko02,klyachko08}, this important job was finally done by Adán Cabello, Simone Severini, and Andreas Winter in 2010 \cite{cabello10} (see also \cite{cabello14}), where they unified contextuality with the notion of nonlocality and provided elegant algorithms to calculate the relevant properties\footnote{We shall discuss this work in section \ref{sec:cabelloseveriniwinter}.}. But they did not provide the sorely needed in-depth discussion of the definition of contextuality.
	
	A clear and well-motivated exposition of the (now) accepted definition of contextuality was finally done by Samson Abramsky and Adam Brandenburger in 2011 \cite{abramsky11}, where they based their definition on the marginal problem and the Fine theorem (these concepts are discussed in section \ref{sec:marginalproblem}). Unfortunately, the authors have chosen to write this paper in the language of category theory, making it inaccessible to most physicists. A clearer explanation of some of their concepts can be found in \cite{kurzynski12,fritz11,chaves12}.

	Now, we shall present this definition and argue that it must be the ``correct'' one. Of course, this statement implies that the definitions discussed in section \ref{sec:spekkenscontextuality} were \emph{wrong}. In fact, it is quite a surprise that the correct definition took 44 years to appear, since the notion was first discussed in \cite{bell66}. One could also argue\footnote{But we're not going to.} that it should be considered 50 years \cite{specker60,seevinck11}, or even 148 years \cite{boole62}.

	This language is purposefully provocative and should be considered somewhat tongue-in-cheek, as it does not make sense, strictly speaking, to talk about correct or incorrect definitions. We do believe, however, that the new definition is a significant improvement over the old ones, as it is already proving itself more fruitful.

	To begin, let's start with our muse, the definition of locality:
	\begin{definition}[Locality]
	A set of probability distributions $p(a_i,b_j|A_i,B_j)$, where $A$ and $B$ refer to independent systems, is local if there exist response functions $\xi_{a_i|A_i}(\lambda), \xi_{b_j|B_j}(\lambda)$ and a probability distribution $\mu(\lambda)$ such that
	\begin{equation}\label{eq:locality} p(a_i,b_j|A_i,B_j) = \int_\Lambda \dint\lambda\, \mu(\lambda)\xi_{a_i|A_i}(\lambda) \xi_{b_j|B_j}(\lambda) \end{equation}
	\end{definition}
	This definition was motivated by the belief that ``correlations cry out for explanation'' \cite{bell81} or, to put it differently\footnote{As Marco Túlio does \cite{quintino12}.}, ``for those who know $\lambda$ there are no correlations'', which could be interpreted as\footnote{Of course, we demand that $p(x|\mathcal X) = \int_\Lambda \dint\lambda\,\mu(\lambda)\xi_{x|\mathcal X}(\lambda)$ for every $x,\mathcal X$, if anything is to make sense.}
	\begin{equation}\label{eq:fatorabilidade} \xi_{a_i,b_j|A_i,B_j}(\lambda) = \xi_{a_i|A_i}(\lambda) \xi_{b_j|B_j}(\lambda) \end{equation}
	Note that equation \eqref{eq:fatorabilidade} can in fact be proved\footnote{This proof seems to be part of the folklore.} (and, consequently, \eqref{eq:locality}) if we assume that $\xi_{a_i,b_j|A_i,B_j}(\lambda)$ is deterministic and non-signalling:
	\begin{definition}[No-signalling]\label{def:nonsignalling}
	We say that a set of probability distributions is non-signalling if for every $A_i$ the marginal
	\[p(a_i|A_i,B_j) = \sum_{b_j} p(a_i,b_j|A_i,B_j)\] does not depend on $B_j$, where $A$ and $B$ refer to independent systems.
	\end{definition}
	\begin{lemma}\label{lem:fatorabilidade}
	Every deterministic probability distribution $p(a_i,b_j|A_i,B_j)$ is factorizable, \ie, there exist probability distributions $p(a_i|A_i,B_j)$ and $p(b_j|A_i,B_j)$ such that \[p(a_i,b_j|A_i,B_j) = p(a_i|A_i,B_j)p(b_j|A_i,B_j).\]
	\end{lemma}
	\begin{proof}
	Define the marginals $p(a_i|A_i,B_j) = \sum_{b_j'} p(a_i,b_j'|A_i,B_j)$ and $p(b_j|A_i,B_j) = \sum_{a_i'} p(a_i',b_j|A_i,B_j)$. Then 
	\begin{align*} 
	p(a_i|A_i,B_j)p(b_j|A_i,B_j) &= \sum_{a_i',b_j'} p(a_i,b_j'|A_i,B_j)p(a_i',b_j|A_i,B_j) \\
				     &= \sum_{a_i',b_j'} \delta_{a_ia_i'}\delta_{b_jb_j'}\de{p(a_i,b_j|A_i,B_j)}^2 \\
				     &= p(a_i,b_j|A_i,B_j),
	\end{align*}
	since $p(a_i,b_j|A_i,B_j)$ is nonzero for a single pair $a_i,b_j$.
	\end{proof}
	\begin{theorem}
	If a set of probability distributions is deterministic and non-signalling, then it is local.
	\end{theorem}
	\begin{proof}
	Define $p(a_i,b_j|A_i,B_j) = \xi_{a_i,b_j|A_i,B_j}(\lambda)$. Applying lemma \ref{lem:fatorabilidade} and definition \ref{def:nonsignalling}, we have equation \eqref{eq:fatorabilidade}, which implies locality.
	\end{proof}
	Therefore, if one believes in determinism and (relativity-enforced) non-signalling, there's quite a good justification for the \textit{factorizability} condition expressed in equation \eqref{eq:fatorabilidade}, and therefore for Bell's definition of locality. But we see that determinism is just a possible justification for it, and not at all a necessary assumption for talking about locality. Without determinism, some valid justifications for factorizability are
\begin{enumerate}
	\item Classical theories are factorizable, as can be seen by the Gelfand-Naimark theorem \cite{strocchi08}. After all, the motivation for looking for a ontological theory in the first place was to recover our classical intuition in a quantum setting.
	\item We don't demand that $\lambda$ gives us deterministic answers; but without factorizability then $\lambda$ does not even explain correlations. And if $\lambda$ does not even explain correlations, why bother with it?
	\item A set of probability distributions admits a joint probability distribution if and only if they are factorizable, as proven by the Fine theorem \cite{fine82} (see our theorem \ref{teo:fine}): 
\end{enumerate}

	In fact, in our opinion the best possible justification for the assumption of factorizability is the Fine theorem, as the existence of a global probability distribution is very appealing on physical grounds. It also shows that the assumption of factorizability implies determinism, so there is in fact nothing else to justify. 

	The Fine theorem shall be our final aim when adapting this discussion to contextuality. We start, however, from humbler considerations. First notice that definition of no-signalling (definition \ref{def:nonsignalling}) does not require any idle talk about relativity, if we do not require that $A_i$ and $B_j$ belong to separate parties, just that they can be jointly measured (which is the only prerequisite for talking about their joint distribution). If we rewrite it like this, we end up with a version of Bell's definition of contextuality for probability distributions:
	\begin{definition}[\textcolor{red}{Wrong}]\label{def:contextualitywrong}
	We say that a set of probability distributions is noncontextual if for every $A_i$ the marginal
	\[p(a_i|A_i,A_j) = \sum_{a_j} p(a_i,a_j|A_i,A_j)\] does not depend on $A_j$.
	\end{definition}
	It is also fair to consider this definition to be a version of Spekkens' definition of measurement contextuality for probability distributions. But we know that this definition is not enough for locality: if we do not also assume factorizability -- or determinism -- all hell breaks loose: it becomes trivial to construct models that violate locality. In fact, notice that the trivial ontological model discussed in section \ref{sec:naive} -- which is neither factorizable nor deterministic -- violates locality; and that by this limited definition of contextuality it would be considered noncontextual, a truly unacceptable proposition. That is why we call these definitions wrong: they are just a generalization of no-signalling. Certainly desirable and useful, but not the whole story.

	Following \cite{ramanathan12}, we shall call this generalized no-sigalling property \emph{no-disturbance}:

	\begin{definition}[No-disturbance]\label{def:nodisturbance}
	We say that a set of probability distributions respects no-disturbance if for every $A_i$ the marginal
	\[p(a_i|A_i,A_j) = \sum_{a_j} p(a_i,a_j|A_i,A_j)\] does not depend on $A_j$.
	\end{definition}

	The full definition of noncontextuality follows from joining no-dis\-tur\-ban\-ce with factorizability, mirroring the definition of locality:

	\begin{definition}[Contextuality]
	A set of probability distributions $p(a_i,a_j|A_i,A_j)$ is noncontextual if there is a probability distribution $\mu(\lambda)$ and response functions $\xi_{a_i|A_i}(\lambda)$ such that
	\[ p(a_i,a_j|A_i,A_j) = \int_\Lambda \dint\lambda\,\mu(\lambda)\xi_{a_i|A_i}(\lambda)\xi_{a_j|A_j}(\lambda).\]
	\end{definition}
	Note that this definition is not quite revolutionary, as most works on contextuality only considered \emph{deterministic} noncontextuality. Its great value comes from the clarity it provides, particularly on the issue of non-de\-ter\-mi\-nis\-tic models: it becomes immediately obvious how to allow for nondeterminism without trivializing our requirements, and shows that the discussion on whether the response functions associated to effects must be deterministic is completely irrelevant. In fact, with it we can ask whether POVMs can be useful to observe contextuality, a question hitherto unexplored.

	Furthermore, it should be clear that this definition is exactly the same as the definition of locality, modulo the restriction that $A_i$ and $A_j$ are observables on separate subsystems; so locality is just a (interesting) particular case of noncontextuality\footnote{Note that even when one is only interested in tests of noncontextuality, this particular case is quite useful, since spatial separation is a good experimental technique to ensure compatibility of the measured observables.}. We shall therefore only talk about contextuality and noncontextuality, restricting our attention to locality if interesting. Notice also that although we only talk about pairs of jointly measurable observables, this definition is naturally extended for sets of any (finite) size, with a corresponding extension to multipartite locality.

	To complete the discussion of contextuality, the only thing lacking is a Fine theorem for noncontextual distributions. By now it should be obvious that it must exist, but we prefer to stop here and establish some notation and formalize what we already have, in order to be able to give a more precise statement. The theorem shall be proved in the next section.

\section{The marginal problem}\label{sec:marginalproblem}

	This notation and definitions are from \cite{abramsky11,fritz11,chaves12}, and are just a formalization of the discussion of the previous section. 

	Let $\mathcal{X} = \{X_0,\ldots,X_{k-1}\}$ be a set of random variables.

\begin{definition}[Marginal scenario]
A marginal scenario $\mathcal{C}$ is a collection $\mathcal{C} = \{C_0,\ldots,C_{n-1}\}$ of subsets $C_i \subseteq \mathcal{X}$ such that $C' \subseteq C_i$ implies $C' \in \mathcal{C}$.
\end{definition}

	The motivation behind this definition is to define which subsets of $\mathcal X$ can be measured simultaneously, in order to actually measure them and generate the probability distributions that will be tested for compatibility. We call the subsets $C_i$ \textit{contexts}, and $\marginal$ is the set of all measurable contexts. Note that in quantum mechanics $\marginal$ will be precisely the subsets of $\mathcal X$ that commute pairwise.

	An interesting particular case is that of Bell scenarios:
\begin{definition}[Bell scenario]
We say that a marginal scenario $\marginal$ is a (bipartite) Bell scenario when there is a partition of $\mathcal X$ into two sets $A = \{A_i\}$ and $B = \{B_i\}$ such that each context $C_i \in \marginal$ contains at most one observable from $A$ and one observable from $B$. The multipartite case can be defined in the same fashion.
\end{definition}
	Note that each context will be of the form $C_i = \{A_k,B_l\}$ (plus the sin\-gle\-tons $C_i = \{A_k\}$ or $C_i = \{B_l\}$), so we can always implement this scenario in quantum mechanics via a tensor product structure, \ie, by defining observables $A_i = \tilde{A_i} \otimes \id$ and $B_j = \id \otimes \tilde{B_j}$. It then becomes possible to consider $A$ and $B$ as independent, spatially separated quantum systems, and to make the measurement of $A_i$ and $B_j$ with a space-like separation. In this way, each choice of context can be justified by an assumption of causality. A natural example of a Bell scenario is the CHSH scenario\footnote{Which shall be discussed in section \ref{sec:chsh}.}, where 
	\[\marginal_{CHSH} = \{\{A_0\},\{A_1\},\{B_0\},\{B_1\},\{A_0,B_0\},\{A_0,B_1\},\{A_1,B_0\},\{A_1,B_1\}\}.\] 
	This definition is only interesting because there are marginal scenarios where one cannot justify the choice of context by arguing that they are measurements on independent subsystems. This scenario is useful for proofs of contextuality, not nonlocality. An interesting example of it is the Klyachko scenario\footnote{Which shall be discussed in section \ref{sec:klyachko}.}, where 
	\begin{multline*}
	\marginal_{K} = \{ \{A_0\},\{A_1\},\{A_2\},\{A_3\},\{A_4\}, \\ 
	\{A_0,A_1\},\{A_1,A_2\},\{A_2,A_3\},\{A_3,A_4\},\{A_4,A_0\}\}.
	\end{multline*}
	There is still a third interesting case, a \textit{partial} Bell scenario, where it is still natural to define two subsystems, but we can't justify \textit{all} the contexts by an assumption of causality, only some. A trivial example of such a scenario would be joining $\marginal_K$ with an observable $B_0$ that can be in every context of $\marginal_K$. A more interesting example would joining $\marginal_K$ with a copy of itself $\marginal_K'$, where we assume that every observable in the first scenario can be in a context with every observable in the second scenario. In this case, we can have violations of both noncontextuality and locality, with some violations of noncontextuality not implying a violation of locality. But we are getting ahead of ourselves; to properly define what we mean by a violation we need a method of assigning probabilities to marginal scenarios and a definition of noncontextuality and locality within this formalism.
\begin{definition}[Marginal model\footnote{Alternative names for marginal models are \textit{behaviour} \cite{tsirelson93} and \textit{box} \cite{barrett05}.}]\label{def:marginalmodel}
A marginal model $\mathcal{C}p$ of a marginal scenario $\mathcal{C}$ is an assignment of probability distributions $C_i \mapsto p(c_i|C_i)$ such that\footnote{With a slight abuse of notation.} \[C_i \subseteq C_j \Rightarrow \sum_{c_j\setminus c_i} p(c_j|C_j) = p(c_i|C_i)\]
\end{definition}
	That is, for every context $C_i$ we assign a probability distribution $p(c_i|C_i)$, where $c_i$ is a vector of possible answers to the random variables contained within $C_i$. Note that this rather minimal compatibility condition on the marginals of the probability distributions is just the no-disturbance condition (definition \ref{def:nodisturbance}). We chose to demand it because marginal models that violate no-disturbance are trivially contextual, and we want to restrict our attention to the interesting cases.

	The reason for this definition is that we can assign these probability distributions to the context in an empirical manner -- for example, from quantum mechanical measurements -- opening up the possibility of a experimental test of locality and noncontextuality.

	With the definition of a marginal model, it becomes possible to state the definition of contextuality within this formalism:
\begin{definition}[Contextuality]\label{def:contextualityfinal}
A marginal model is noncontextual if there are response functions $\xi_{x_i|X_i}(\lambda)$ and a probability distribution $\mu(\lambda)$ such that for every $C_i \in \mathcal{C}$
	\[ p(c_i|C_i) = \int_\Lambda \dint\lambda\,\mu(\lambda) \prod_{x_n \in \, c_i} \xi_{x_n|X_n}(\lambda) \]
\end{definition}
	Naturally, we say that a marginal model is contextual if it is not noncontextual. Note that the definition of locality is the same, with the restriction that the marginal scenario is actually a Bell scenario; analogously, we say that a marginal model is nonlocal if it is not local.
	
	Having definition \ref{def:contextualityfinal}, we can state and prove the generalized Fine theorem that motivates it\footnote{It was first considered by Liang \etal \cite{liang11} and proved by Abramsky \etal \cite{abramsky11}.}\footnote{In fact, the motivation is so strong that some prefer to consider definition \ref{def:contextualityfinal} as defining ``objective reality'' instead of noncontextuality \cite{kurzynski12}. Although we agree that this interpretation is not inappropriate, we prefer to avoid such dramatic terms.}:
\begin{theorem}[Fine \cite{fine82,liang11,abramsky11}]\label{teo:fine}
A marginal model $\mathcal{C}$ is noncontextual iff there exists a probability distribution $p(x|\mathcal{X})$ such that for every $C_i \in \mathcal{C}$ \[ p(c_i|C_i) = \sum_{x\setminus c_i} p(x|\mathcal{X}) \]
\end{theorem}
\begin{proof}
\mbox{ }
\begin{description}
	\item[$\Rightarrow$] By noncontextuality, there are response functions $\xi_{x_i|X_i}(\lambda)$ and a probability distribution $\mu(\lambda)$ such that for every $C_i \in \mathcal{C}$
	\[ p(c_i|C_i) = \int_\Lambda \dint\lambda\,\mu(\lambda) \prod_{x_n \in \, c_i} \xi_{x_n|X_n}(\lambda). \]
	Define 
	\[ p(x|\mathcal{X}) = \int_\Lambda \dint\lambda\,\mu(\lambda)\prod_{x_n \in\, x}\xi_{x_n|X_n}(\lambda) \]
	Then any marginal $p(c_i|C_i)$ is given by 
	\begin{align*}
	p(c_i|C_i) &= \sum_{x\setminus c_i} p(x|\mathcal{X}) \\
		   &= \int_\Lambda \dint\lambda\,\mu(\lambda) \sum_{x\setminus c_i}\prod_{x_n \in\, x}\xi_{x_n|X_n}(\lambda) \\
		   &= \int_\Lambda \dint\lambda\,\mu(\lambda)\prod_{x_n \in\, c_i} \xi_{x_n|X_n}(\lambda)
	\end{align*}
	\item[$\Leftarrow$] Every probability distribution $p(x|\mathcal X)$ can be written as a convex combination of deterministic points, so let
	\[ p(x|\mathcal X) = \int_\Lambda\dint\lambda\,\mu(\lambda) \xi_{x|\mathcal X}(\lambda). \]
	Since deterministic probability distributions are factorizable (lemma \ref{lem:fatorabilidade}), we can write
	\[ p(x|\mathcal X) = \int_\Lambda\dint\lambda\,\mu(\lambda) \prod_{x_n\in \mathcal X} \xi_{x_n|X_n}(\lambda). \]
	By assumption, $p(c_i|C_i) = \sum_{x\setminus c_i} p(x|\mathcal{X})$, so 
	\begin{align*}
	p(c_i|C_i) &= \int_\Lambda\dint\lambda\,\mu(\lambda) \sum_{x\setminus c_i} \prod_{x_n\in \mathcal X} \xi_{x_n|X_n}(\lambda) \\
		   &= \int_\Lambda \dint\lambda\,\mu(\lambda) \prod_{x_n \in \, c_i} \xi_{x_n|X_n}(\lambda).
	\end{align*}
\end{description}
\end{proof}
	Note that in the proof of the Fine theorem we can choose the response functions $\xi_{x_n|X_n}(\lambda)$ to be always deterministic, so
	\begin{corollary}\label{cor:fine}
	A marginal model is noncontextual if and only if there are deterministic response functions $\xi_{x_i|X_i}(\lambda)$ and a probability distribution $\mu(\lambda)$ such that for every $C_i \in \mathcal{C}$
	\[ p(c_i|C_i) = \int_\Lambda \dint\lambda\,\mu(\lambda) \prod_{x_n \in \, c_i} \xi_{x_n|X_n}(\lambda) \]
	\end{corollary}
	This corollary can be viewed as an alternative (equivalent) definition of noncontextuality.

	Now, we can finally state the problem of separating between classical and quantum:
	
\begin{problem}[Marginal problem]\label{pro:marginal}
How to decide whether a given marginal model is noncontextual or contextual?
\end{problem}

	This formulation of the problem makes its mathematical treatment much easier, since there is extensive literature (and software) on solving the marginal problem. But perhaps its greatest contribution is ending the debate on whether contextuality can or not be observed in a laboratory: one measures a marginal model, and then it is just a mathematical question whether it is contextual or not. The ``finite-precision'' \cite{meyer99,kent99} loophole is just not relevant in this formulation, as the set of contextual marginal models has non-empty interior.

\section{A first example}\label{sec:os}

	If we only have two random variables, there's nothing interesting to be done, since either we already have the global distribution, or we can generate it simply by defining\footnote{A moment's thought will convince you that if the marginal scenario contains only the singletons $X_n$, we can always do this and prove that it is noncontextual.} $p(x_0,x_1|X_0,X_1) = p(x_0|X_0)p(x_1|X_1)$, so the simplest nontrivial scenario must contain at least three random variables. In fact, there is a nice little example of it, taken from \cite{liang11}, which took it from Specker's parable of the over-protective seer, that can be found in \cite{specker60,seevinck11}. In it, we have three binary random variables $X_0$, $X_1$, and $X_2$ that are measured pairwise, and found to be always anti-correlated. Formalizing it, the marginal scenario is 
	\[ \mathcal{OS} = \{ \{X_0\},\{X_1\},\{X_2\},\{X_0,X_1\},\{X_1,X_2\},\{X_2,X_0\}\},\]
	and its marginal model $\mathcal{OS}p$ is (with a slight abuse of notation)
	\begin{multline}\label{eq:ospruim}
	\mathcal{OS}p = ( p(x_0|X_0),p(x_1|X_1),p(x_2|X_2),\\p(x_0,x_1|X_0,X_1),p(x_1,x_2|X_1,X_2),p(x_2,x_0|X_2,X_0)),
	\end{multline}
	which for convenience we arrange in the following tables:
\begingroup\renewcommand*{\arraystretch}{1.5}

	\[ \begin{matrix}
				&X_0 	&	X_1	&	X_2 \\ \cline{2-4}
			p(+) & \frac{1}{2} & \frac{1}{2} & \frac{1}{2} \\
			p(-) & \frac{1}{2} & \frac{1}{2} & \frac{1}{2} \\
	\end{matrix}\quad\quad\begin{matrix} &X_0,X_1&X_1,X_2&X_2,X_0 \\ \cline{2-4}
			p(+,+)  &0 & 0 & 0 \\ 
			p(+,-)  &\frac{1}{2} & \frac{1}{2} & \frac{1}{2} \\
			p(-,+)  &\frac{1}{2} & \frac{1}{2} & \frac{1}{2} \\
			p(-,-)  &0 & 0 & 0 \\
	\end{matrix} \]
\endgroup
	To see that this marginal model is contextual, we shall use the Fine theorem (theorem \ref{teo:fine}), as in \cite{liang11}, by showing that there can be no global probability distribution $p(x|X)$ with these marginals.
	\begin{theorem}\label{teo:contextualdirect}
	The marginal model $\mathcal{OS}p$ is contextual.
	\end{theorem}
	\begin{proof}
	$p(+,+|X_0,X_1)=0$ implies that both $p(+,+,+|X_0,X_1,X_2)$ and \\$p(+,+,-|X_0,X_1,X_2)$ must be zero. Proceeding in this way with the other marginals, we can show that all $p(x_0,x_1,x_2|X_0,X_1,X_2)$ are zero, an absurd. So there is no global probability distribution and by theorem \ref{teo:fine} $\mathcal{OS}p$ is contextual.
	\end{proof}
	An interesting question is then whether this contextual marginal model can be used as a proof of contextuality for quantum mechanics. Unfortunately this is not the case, as it requires all three products of observables $X_iX_j$ to be measurable; in quantum mechanics this means that they must commute, and therefore the observable $X_0X_1X_2$ must be measurable, giving rise to the joint probability distribution that must not exist. A marginal scenario with three random variables that averts this problem is 
	\[ \mathcal{V} = \{ \{X_0\},\{X_1\},\{X_2\},\{X_0,X_1\},\{X_1,X_2\}\},\]
	since it is perfectly possible that $X_1$ commutes with both $X_0$ and $X_2$, but $X_0$ and $X_2$ does not commute. But this marginal scenario is even more trivial than the previous one, since there is always a noncontextual marginal model for it\footnote{Since we can just define $p(x_0,x_1,x_2|X_0,X_1,X_2)= p(x_0,x_1|X_0,X_1)p(x_1,x_2|X_1,X_2)/p(x_1|X_1)$.}. As these two are the only nontrivial marginal scenario with three random variables, we must have at least four random variables if we want a contextual marginal model realizable within quantum mechanics, and in fact there exists one. To be able to explore it, though, we need a bit more structure, since a direct proof of contextuality \textit{à la} theorem \ref{teo:contextualdirect} can be done only for the simplest cases. In the next section, we shall develop a general algorithm to decide whether a given marginal model is contextual or not.

\section{Boole inequalities}

	\epigraph{When satisfied they indicate that the data \emph{may} have, when not satisfied they indicate that the data \emph{cannot} have, resulted from actual observation}{George Boole \cite{boole62}}

	To be able to solve problem \ref{pro:marginal}, we shall first take a step back and examine its geometry. We shall see that the sets of marginal models are convex polytopes, and these can be described by a finite set of linear inequalities, and so the question of whether a given marginal model is contextual or not is reduced to checking if it satisfies all the inequalities for its marginal scenario. This can be done efficiently, but with two caveats: obtaining the inequalities for a given scenario is a difficult problem (albeit one that can be done by software), and the number of inequalities for a marginal scenario may increase exponentially with the number of contexts\footnote{As in the example of section \ref{sec:nciclo}.}.

	In this section we shall need a number of basic results in convex geometry, which we shall make no attempt to prove. Instead, we refer the interested reader to the excellent book ``Lectures on Polytopes'' \cite{ziegler94}.

	\subsection{Sets of marginal models}

	There are for now two sets of marginal models that interests us: the set of all marginal models, and the set of noncontextual marginal models. We shall see that both are convex polytopes.
	\begin{definition}[Convex polytope]
	A convex polytope is a bounded intersection of closed halfspaces.
	\end{definition}
	\begin{theorem}
	The set of all marginal models for a given marginal scenario is a convex polytope.
	\end{theorem}	
	\begin{proof}
	Consider the marginal scenario \[ \marginal = \{C_1,\ldots,C_N\}, \]and a marginal model\[\marginal p = (p(c_1|C_1),\ldots,p(c_N|C_N)).\] The fact that each $p(c_i|C_i)$ is a probability distribution is encoded by the linear inequalities\footnote{Remember that the equality $x = k$ is just the combination of the inequalities $x \le k$ and $x \ge k$.} $p(c_i|C_i) \ge 0$ and $\sum_{c_i}p(c_i|C_i) = 1$, and the fact that this set of probability distributions is a marginal model is encoded by the no-disturbance condition expressed in the definition \ref{def:nodisturbance}, which is just another set of linear inequalities. It remains to show that the set is bounded, but this follows from the fact that each element of $\marginal p$ belongs to $[0,1]$.
	\end{proof}
	We shall call the set of all marginal models the \textit{no-disturbance polytope}.

	To see that the set of noncontextual marginal models is also a convex polytope, it is easier to use another equivalent\footnote{The proof of their equivalence is the famous Minkowski-Weyl theorem.} definition of convex polytopes:
	\begin{definition}[Convex polytope]
	A convex polytope is the convex hull of a finite set of points in some $\R^n$.	
	\end{definition}
	\begin{theorem}
	The set of all noncontextual marginal models for a given marginal is a convex polytope.
	\end{theorem}	
	\begin{proof}
	Consider the marginal scenario \[ \marginal = \{C_1,\ldots,C_N\}, \]and a marginal model \[\marginal p = (p(c_1|C_1),\ldots,p(c_N|C_N)).\]
	By the corollary \ref{cor:fine} of the Fine theorem \ref{teo:fine}, there is a probability distribution $\mu(\lambda)$ and deterministic response functions $\xi_{x_n|X_n}(\lambda)$ such that \[ p(c_i|C_i) = \int_\Lambda \dint\lambda\,\mu(\lambda) \prod_{x_n \in \, c_i} \xi_{x_n|X_n}(\lambda), \]
	and so \[\marginal p = \int_\Lambda \dint\lambda\,\mu(\lambda) \de{\prod_{x_n \in \, c_1} \xi_{x_n|X_n}(\lambda),\ldots,\prod_{x_n \in \, c_N} \xi_{x_n|X_n}(\lambda)},\]
	that is, $\marginal p$ is a convex combination of the points \[\de{\prod_{x_n \in \, c_1} \xi_{x_n|X_n}(\lambda),\ldots,\prod_{x_n \in \, c_N} \xi_{x_n|X_n}(\lambda)}.\] Since the response functions are deterministic and we are dealing with a finite number of dichotomic random variables, the number of different points is finite, and so a marginal model is the convex combination of a finite number of points.
	\end{proof}
	Analogously, the set of all noncontextual marginal models shall be called the \textit{noncontextual polytope}.

	As a consequence of this proof, we see that the vertices of the noncontextual polytope are simply the deterministic probability distributions for the outcomes of each context, and as such they are trivial to find. What we want to do, then, is from this list of vertices obtain the linear inequalities that describe the noncontextual polytopes. This is a classical problem in convex geometry, and there are plenty of algorithms and software for solving it. Here we shall use the reverse search algorithm, due to Avis and Fukuda \cite{avis92}, as implemented in the software \texttt{lrs} \cite{lrs}. Following Itamar Pitowsky, we call these \textit{Boole inequalities}.

	Before exploring them, we need a refinement in our representation of marginal models.

	\subsection{Representing marginal models}

	When writing down a marginal model, such as \eqref{eq:ospruim}, one immediately notices that it has a lot of redundancies. First of all, the joint probability distributions of a context completely determines its marginals, since a marginal model respects no-disturbance by definition. Furthermore, for each context there is one parameter that is already determined by normalization, and finally each random variable is usually shared by two or more contexts, so the joint probability distributions of different contexts are not independent, as they might share some marginals.

	All these reasons motivates us to find another representation of a marginal model, that already incorporates normalization and no-disturbance. When using only dichotomic random variables (as we shall do in this thesis), the best representation is via the expectation value of each context, as they contain all the information of a marginal model with no redundancies.

	\begin{theorem}
	For dichotomic random variables, a marginal model can be represented by the expectation values of all contexts with no redundancies.
	\end{theorem}
	\begin{proof}
	To check that, it is enough to see that all the information present on the marginal model is preserved when it is translated into expected values, \ie, there is a (linear) invertible transformation between a marginal model and a vector of all the allowed expected values. Consider, for instance, the joint probability distribution for the context $\{X_0,X_1\}$. The transformation is
	\begin{equation}\label{eq:representationmarginal}\begin{pmatrix}
			1 & 1 & 1 & 1 \\
			1 & -1 & 1 & -1 \\
			1 & 1 & -1 & -1 \\
			1 & -1 & -1 & 1 \\
	\end{pmatrix} \begin{pmatrix}
			p(+,+|X_0,X_1) \\
			p(+,-|X_0,X_1) \\
			p(-,+|X_0,X_1) \\
			p(-,-|X_0,X_1) \\
	\end{pmatrix} = \begin{pmatrix}
			1 \\
			\mean{X_1}\\
			\mean{X_0} \\
			\mean{X_0X_1} \\
	\end{pmatrix},
	\end{equation}
	and inversibility comes from the fact that the matrix is proportional to its inverse. 

	The proof for contexts with more than two random variables comes from noticing that the matrix which does the linear transformation is a Hadamard matrix\footnote{Thanks to Daniel Jonathan for pointing this out.}. Specifically, the transformation for $n$ random variables can be recursively defined as follows: Let \[H_1 = \begin{pmatrix}
			1 & 1 \\
			1 & -1 \\
	\end{pmatrix},\]
	and define $H_n = H_1 \otimes H_{n-1}$. Then it is easy to check that $H_n$ is always self-adjoint and $H_n^2 = 2^n \id$. Furthermore, if the vector of probabilities is ordered in the obvious binary way, the vector of expected values will have a corresponding order, \ie, its $k$th element will be $\mean{X_0^{a_0}X_1^{a_1}\ldots X_{n-1}^{a_{n-1}}}$, where $a_0a_1\ldots a_{n-1}$ is the binary expansion of $k$.
	\end{proof}

	As this representation already assumes normalization and no-disturbance, the only information that it lacks is positivity. Since positivity does not reduce the number of dimensions, it is not possible to find a representation that already assumes it. Instead, one enforces it via the inequalities 
	\begin{subequations}\label{eq:positivity}	
	\begin{align} 
		\label{eq:um}		4p(+,+|X_0,X_1) = 1 + \mean{X_0} + \mean{X_1} + \mean{X_0X_1} \ge 0 \\
		\label{eq:dois}		4p(-,+|X_0,X_1) = 1 - \mean{X_0} + \mean{X_1} - \mean{X_0X_1} \ge 0 \\
		\label{eq:tres}		4p(+,-|X_0,X_1) = 1 + \mean{X_0} - \mean{X_1} - \mean{X_0X_1} \ge 0 \\
		\label{eq:quatro}	4p(-,-|X_0,X_1) = 1 - \mean{X_0} - \mean{X_1} + \mean{X_0X_1} \ge 0
	\end{align}
	\end{subequations}
	which are obtained by inverting transformation \eqref{eq:representationmarginal}.

	Using this representation also gives us some notational convenience: since we have one expected value for each context, we can define a marginal model simply by assigning one expected value for each context in a marginal scenario. For example, the marginal model for the marginal scenario 
	\[ \mathcal{OS} = \{ \{X_0\},\{X_1\},\{X_2\},\{X_0,X_1\},\{X_1,X_2\},\{X_2,X_0\}\},\]
	originally written as \eqref{eq:ospruim}, shall be 
	\begin{equation}\label{eq:ospgeral}
	\mathcal{OS}p = (\mean{X_0},\mean{X_1},\mean{X_2},\mean{X_0X_1},\mean{X_1X_2},\mean{X_2X_0}),
	\end{equation}
	which is easily calculated as
	\begin{equation}\label{eq:osp}
	\mathcal{OS}p = (0, 0, 0, -1, -1, -1).
	\end{equation}
	Another advantage of this representation is that we can easily see which statistics that indicate correlations between random variables, such as $\mean{X_iX_j}$, and which only talk about individual systems, such as $\mean{X_i}$. We shall see that it is quite common to study inequalities that only take into account correlations between random variables\footnote{In fact, only these shall be studied in this thesis.}: these are called \textit{full-correlation} inequalities. When talking about contexts with more than two random variables, this name is applied only to inequalities that take into account the largest possible contexts.

\subsection{\texorpdfstring{The noncontextual polytope for $\mathcal{OS}$}{The noncontextual polytope for OS}}

	Now that we have a good representation, we can discuss the first example of Boole inequalities. We shall obtain them for the marginal scenario $\mathcal{OS}$. The first thing we need are the vertices of the noncontextual polytope, which are simply the $2^3$ deterministic assignments $\pm1$ to each random variable $\mean{X_i}$. Written in the ordering given by equation \eqref{eq:ospgeral}, they are
	\begin{align*}
	(+,+,+,+,+,+) & & (-,+,+,-,+,-) \\ 
	(+,+,-,+,-,-) & & (-,+,-,-,-,+) \\
	(+,-,+,-,-,+) & & (-,-,+,+,-,-) \\
	(+,-,-,-,+,-) & & (-,-,-,+,+,+)
	\end{align*}
	where for clarity we have omitted the ones. Inputting these vertices into \texttt{lrs}\footnote{For those that do not like this kind of proof, we shall obtain these same inequalities in the next section via a parity argument.}, it returns 16 inequalities to us: 12 are the positivity conditions \eqref{eq:positivity} for each pair of random variables, and 4 are the Boole inequalities
	\begin{subequations}\label{eq:ostodas}
	\begin{align}
	\label{eq:osspekkens}
	-\mean{X_0X_1} - \mean{X_1X_2} - \mean{X_2X_0} &\le 1 \\ 
	-\mean{X_0X_1} + \mean{X_1X_2} + \mean{X_2X_0} &\le 1 \\
	+\mean{X_0X_1} - \mean{X_1X_2} + \mean{X_2X_0} &\le 1 \\
	+\mean{X_0X_1} + \mean{X_1X_2} - \mean{X_2X_0} &\le 1
	\end{align}
	\end{subequations}
	The marginal model $\mathcal{OS}p$, equation \eqref{eq:osp}, is then easily seen to violate inequality \eqref{eq:osspekkens}, being thereby contextual.

	Exactly these same inequalities were obtained by Pitowsky using Boole's method \cite{pitowsky94,boole62}.

\begin{figure}[t]
	\centering
	\includegraphics[width=0.58\textwidth]{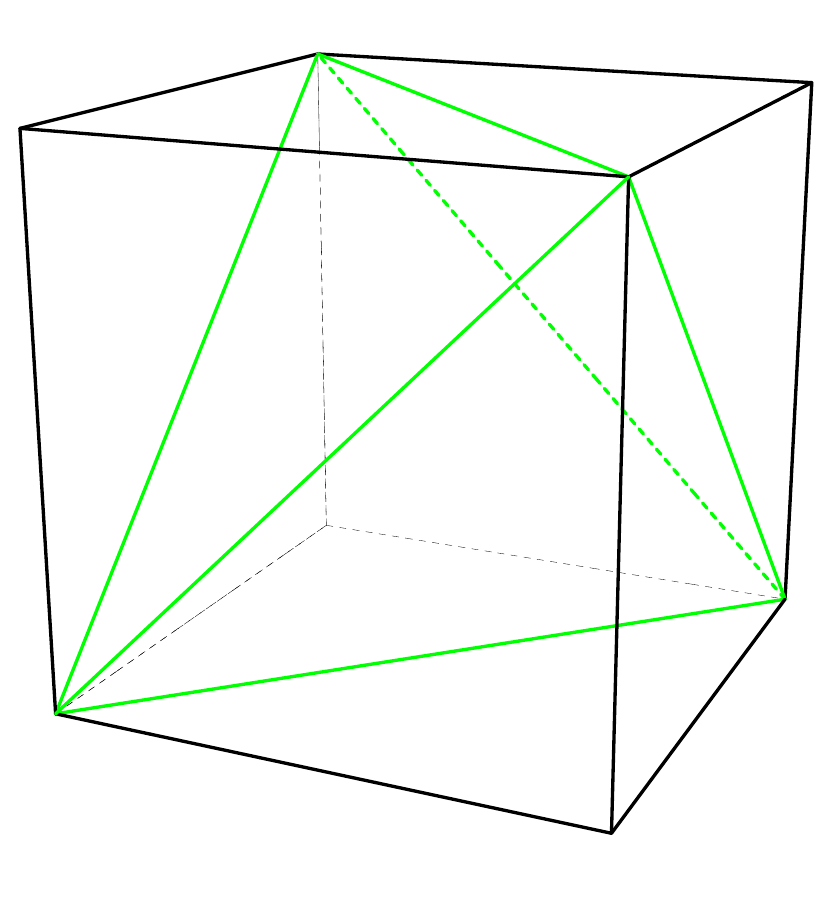}
	\caption{Full correlations parts of the noncontextual polytope -- green tetrahedron -- and no-disturbance polytope -- black cube -- for the marginal scenario $\mathcal{OS}$. Note that this is a projection onto the last three components.}
	\label{fig:3cycle}
\end{figure}

\section{\texorpdfstring{The $n$-cycle}{The n-cycle}}\label{sec:nciclo}

	As we have discussed before, it is not possible to violate the Boole inequalities for the marginal scenario $\mathcal{OS}$ with quantum mechanics. However, there is a natural generalization of this scenario which \emph{does} have a quantum violation. Consider the set of random variables $\mathcal{X} = \{X_0,\ldots,X_{n-1}\}$, and the marginal model $\marginal^n$ formed by considering the singletons ${X_i}$ together with the pairs $\{X_i,X_{i+1}\}$, where naturally the addition is taken modulo $n$. For $n=3$, $\marginal^n$ is the marginal scenario $\mathcal{OS}$ discussed before. For general $n$ this scenario is called the $n$-cycle, as its compatibility\footnote{The graph that has random variables as vertices and edges connect random variables that are in the same context.} graph is a $n$-cycle, as shown in figure \ref{fig:nciclo}.

	\begin{figure}[ht]
	\centering
	\includegraphics[width=1.00\textwidth]{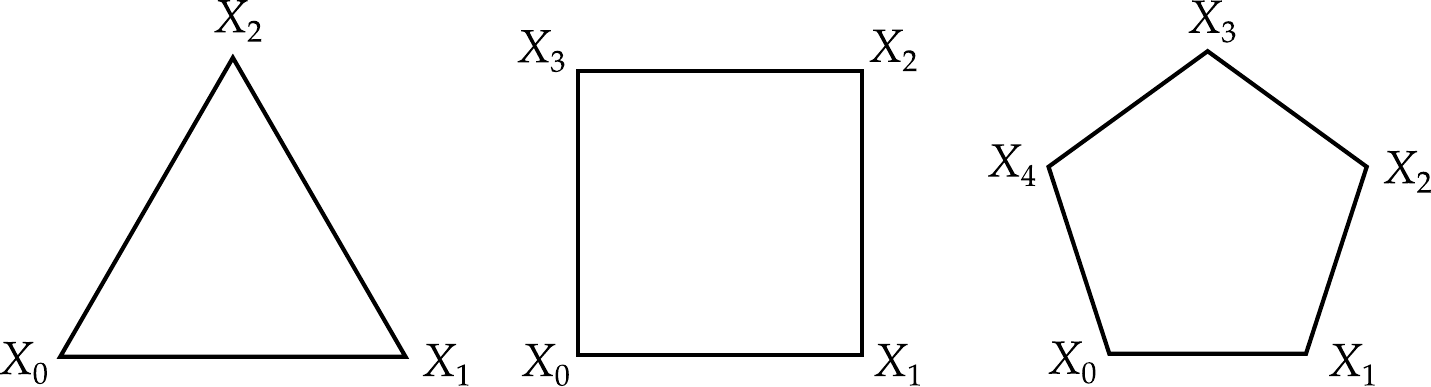}
	\caption{Contexts for the $3$-cycle, $4$-cycle, and $5$-cycle.}
	\label{fig:nciclo}
	\end{figure}

	The $n$-cycle marginal scenario is an old problem that was studied many times. The $3$-cycle\footnote{As we discussed before, in this case the noncontextual polytope coincides with the no-disturbance polytope, and therefore its facets are only the positivity conditions \eqref{eq:positivity}.} was characterized by George Boole in 1862 \cite{boole62,pitowsky94}, who also provided the general algorithm for solving the marginal problem. The $3$-cycle was only studied again almost a hundred years later, by Ernst Specker in 1960 \cite{specker60,seevinck11}, and then by Itamar Pitowsky in 1989 \cite{pitowsky89}. The $4$-cycle was characterized by Arthur Fine in 1982 \cite{fine82}. The $5$-cycle was characterized by Alexander Klyachko in 2002 \cite{klyachko02}. The $n$-cycle for all odd $n$ was studied by Yeong-Cherng Liang, Robert Spekkens, and Howard Wiseman in 2010 \cite{liang11}, and also by Adán Cabello, Simone Severini, and Andreas Winter in the same year \cite{cabello10}. The general $n$-cycle was studied by Rafael Chaves and Tobias Fritz in 2012, who derived entropic inequalities which are necessary but not sufficient for noncontextuality for all $n$ \cite{chaves12,fritz11}. An answer to the general question was conjectured by Cabello \etal in 2012 \cite{cabello12}. It will be given here\footnote{The results of this and the next section are new \cite{araujo12}.}.

	The Boole inequalities for this scenario can be derived from the simple algebraic observation that if $\alpha_i =\pm1$ are the components of a $n$-element vector, then the vector $\beta$ with $n$ components $\beta_i = \alpha_i\alpha_{i+1}$ always has an \emph{even} number of negative components. Therefore, if we define a third vector $\gamma$ with an \emph{odd} number of negative components, then
	\begin{equation}\label{eq:boolealgebrico}
	\prin{\gamma}{\beta} \le n-2,
	\end{equation}
	since to maximize the inner product we should set $\beta = \gamma$, but this would force $\beta$ to have an odd number of negative components, which is impossible. The best we can do then is to switch one of the $-1$ to $+1$, which gives us the desired bound.

	If we now set $\mean{X_i} = \alpha_i$, then $\beta_i = \mean{X_iX_{i+1}}$ is the full-correlation part of the vertices of the noncontextual polytope for this marginal scenario, and inequality \eqref{eq:boolealgebrico} becomes the Boole inequality
	\begin{equation}\label{eq:boolenciclo}
	\chsh_n = \sum_{i=0}^{n-1} \gamma_i \mean{X_iX_{i+1}} \le n-2.
	\end{equation}
	Since these are satisfied by noncontextual vertices, they are also satisfied by the convex combinations of them, and so every noncontextual marginal model respects these inequalities. We claim that these $2^{n-1}$ inequalities are all the Boole inequalities for the $n$-cycle. To prove this, we shall check that these inequalities are actually facets of the noncontextual polytope, and that there are no more Boole inequalities for the $n$-cycle.
	\begin{theorem}
	All inequalities \eqref{eq:boolenciclo} are facets of the noncontextual polytope for the $n$-cycle.
	\end{theorem}
	\begin{proof}
We will check that each Boole inequality (\ref{eq:boolenciclo}) is saturated by $2n$ affinely independent vertices of the noncontextual polytope, that generate an affine subspace of dimension $2n-1$. Note that if we flip the sign of any component $\gamma_i$ of the Boole inequality $\gamma$, then this new vector $\gamma'$ satisfies $\prin{\gamma}{\gamma'} = n-2$ and has an even number of negative components, so we have obtained the full-correlation part of a noncontextual vertex that saturates the Boole inequalities. Since there are two ways of completing the local part of a noncontextual vertex that are consistent with a given full-correlation part and we have $n$ components $\gamma_i$ to flip the sign, in this manner we obtain $2n$ vertices of the noncontextual polytope that saturate the Boole inequality $\gamma$. To check that they are affinely independent is trivial. 
	\end{proof}

	To check that there are no more Boole inequalities, we need first to characterize the contextual vertices of the no-disturbance polytope.

\begin{theorem}\label{teo:nd}
The vertices of the no-disturbance polytope are the $2^n$ noncontextual deterministic marginal models
\begin{equation}\label{eq:nc}
(\langle X_0\rangle,\ldots,\langle X_{n-1}\rangle,\langle
X_0\rangle\langle X_1\rangle,\ldots,\langle X_{n-1}\rangle\langle
X_0\rangle),
\end{equation}
where $\langle X_i\rangle = \pm 1$, together with the $2^{n-1}$ contextual marginal models of the form
\begin{equation}\label{eq:c}
(0,\ldots,0,\langle X_0 X_1\rangle,\ldots,\langle X_{n-1} X_0\rangle),
\end{equation}
where $\langle X_i X_{i+1}\rangle=\pm 1$ such that number of negative components is odd.
\end{theorem}
\begin{proof}
By definition, the vertices of the polytope are given by the intersection of $2n$ independent hyperplanes, \ie, as a unique solution for a set of $2n$ independent linear equations chosen among the $4n$ equations \eqref{eq:positivity}. The above vertices are obtained
 by choosing two equations among \eqref{eq:um}-\eqref{eq:quatro}, for each index $i$. In particular,  contextual vertices are obtained by choosing equations \eqref{eq:um} and \eqref{eq:quatro} for an odd number of indexes $i$ and equations \eqref{eq:dois} and \eqref{eq:tres} for the remaining indexes.

It is straightforward to check that all other possible strategies for obtaining a vertex, \ie, involving the choice of $1,2$ or $3$ equations for each index $i$, give the same set of vertices.\end{proof}

	We now show that by eliminating each contextual vertex of the no-dis\-tur\-ban\-ce polytope we obtain only one noncontextuality inequality. By eliminating all $2^{n-1}$ contextual vertices, we obtain $2^{n-1}$ noncontextuality inequalities and the convex hull of all noncontextual vertices, \ie, the noncontextual polytope.

\begin{lemma}\label{lem:elimination}
Let $C$ be a contextual vertex, and consider the inequality (\ref{eq:boolealgebrico}) with $\gamma = C$. Then the intersection of the half-space $\prin{\gamma}{P} \le n-2$ with the no-disturbance polytope is the convex hull of all vertices but $C$.
\end{lemma}
\begin{proof}
To show that, we shall check that the vertices of the intersection of the half-space $\prin{\gamma}{P} \le n-2$ with the no-disturbance polytope are a subset of the vertices of the no-disturbance polytope. For contradiction, suppose that the intersection generates a new vertex $P'$ that was not a vertex of the no-disturbance polytope. Then $\prin{\gamma}{P'} = n-2$ and, furthermore, $P'$ must lie on an edge connected to $C$, since all the other vertices respect the inequality. Edges of the no-disturbance polytope must saturate $2n-1$ independent positivity conditions \eqref{eq:positivity}, and therefore $P'$ must saturate $2n-1$ inequalities which are a subset of the $2n$ inequalities saturated by the vertex $C$. 

	Let $\beta$ be the full-correlation part of $C$, and $\delta$ the full-correlation part of $P'$. For each $i$, if $\beta_i=+1$, then $C$ saturates \eqref{eq:dois} and \eqref{eq:tres}. If $\beta_i=-1$, $C$ saturates \eqref{eq:um} and \eqref{eq:quatro}. Therefore, for every $i$ but one, let's say, $i_0$, $P'$ must saturate both positivity conditions; but saturating them both implies that $\delta_i = \beta_i$, leaving only $\delta_{i_0}$ free. But if we now demand that $\prin{\gamma}{P'} = n-2$, then $\delta_{i_0} = -\beta_{i_0}$, and therefore $P'$ is just an old noncontextual vertex.
\end{proof}

	To summarize our results: the no-disturbance polytope has $2^n + 2^{n-1}$ vertices, of which $2^n$ are noncontextual and $2^{n-1}$ are contextual. It has $4n$ facets, which are the positivity conditions \eqref{eq:positivity}. The noncontextual polytope has $2^n$ vertices and $4n + 2^{n-1}$ facets.

	\subsection{Quantum violations}

	The Boole inequalities for the $n$-cycle are violated by quantum mechanics for every $n \ge 4$. Since the inequalities for a given $n$ are all equivalent via relabellings, it is enough to violate one of them. For odd $n$, we choose the inequality with all $\gamma_i = -1$. The minimal dimension we need to violate the Boole inequalities is $3$, the state is always $\ket{0}$, and the observables\footnote{These states and observables are from \cite{liang11}.} are $A_k = 2\ketbra{v_k}{v_k}-\id$, where
	\[ \ket{v_k} = (\cos\theta,\sin\theta\cos\phi_k,\sin\theta\sin\phi_k),\]
	where \[\phi_k = \frac{n-1}{n}\pi k\] and 
	\[ \cos^2\theta = \frac{\cos\frac{\pi}{n}}{1+\cos\frac{\pi}{n}}. \]
	Then $\exval{0}{A_kA_{k+1}}{0} = -4\abs{\braket{0}{v_k}}^2+1 = -4\cos^2\theta +1$, and 
	\begin{equation}\label{eq:oddviolation}
	\chsh_n = n\de{4\frac{\cos\frac{\pi}{n}}{1+\cos\frac{\pi}{n}} -1}.
	\end{equation}
	The noncontextual bound is $\chsh_n \le n-2$. This inequality is saturated for $n=3$, and violated for all $n\ge 5$. To see this, it is enough to use some simple algebra and the fact that \[\cos\frac\pi n > 1-\frac{\pi^2}{n^2}\] for all $n$.

	For even $n$, we choose the inequality for which all $\gamma_i=-1$ except for $\gamma_{n-1}=+1$. Dimension $4$ is enough to violate\footnote{We conjecture that this is in fact the minimal dimension. For $n=4$ the proof is well-known.} it for all $n$, with the state
	\[\ket{\psi_-} = \ket{01}-\ket{10},\]
	and the observables\footnote{These states and observables are from \cite{braunstein89}.} $X_k = \tilde X_k \otimes \id$ for even $k$ and $X_k = \id \otimes \tilde X_k$ for odd $k$, where
	\[\tilde X_k = \cos\frac{k\pi}{n} \sigma_x + \sin \frac{k\pi}{n} \sigma_z, \]
	and $\sigma_x,\sigma_z$ are the Pauli matrices.

	We can then check that 
	\[X_kX_{k+1}\ket{\psi_-} = -\cos\frac \pi n \ket{\psi_-}-\sin\frac \pi n \ket{\phi_+}\]
	for every $k$ except $k=n-1$, when 
	\[X_{n-1}X_0\ket{\psi_-} = \cos\frac \pi n \ket{\psi_-}-\sin\frac \pi n \ket{\phi_+}.\]
	Therefore, 
	\begin{equation}\label{eq:evenviolation}
	\chsh_n = n\cos\frac \pi n,
	\end{equation}
	so the noncontextual bound is saturated for $n=2$, and violated for all $n\ge 4$.

	Note that in both the even and odd cases $\lim_{n\to\infty}\chsh_n = n$, the algebraic bound.

	\subsection{The CHSH inequality}\label{sec:chsh}

	The $4$-cycle is actually a Bell scenario, since every observable in the set $\{X_0,X_2\}$ commutes with every observable in the set $\{X_1,X_3\}$. Renaming $A_0 = X_0$, $A_1 = X_2$, $B_0 = X_1$, and $B_1 = X_3$, we have the famous CHSH inequality \cite{chsh69}. 
	\begin{equation}\label{eq:chsh}
	\mean{A_0B_0} + \mean{A_0B_1} + \mean{A_1B_0}-\mean{A_0B_1} \le 2
	\end{equation}
	The maximal quantum violation for it -- its Tsirelson bound \cite{tsirelson80} -- is $2\sqrt{2}$. This inequality was used in countless experimental tests of nonlocality, of which the most famous are the first, by Freedman and Clauser \cite{freedman72}, and Aspect's \cite{aspect82}.

	\subsection{The Klyachko inequality}\label{sec:klyachko}

	The $5$-cycle was studied before by Klyachko \cite{klyachko02}, and the following inequality got his name:
	\begin{equation}\label{eq:klyachko}
	-\mean{X_0X_1} - \mean{X_1X_2} - \mean{X_2X_3} -\mean{X_3X_4} -\mean{X_4X_0} \le 3. 
	\end{equation}
	Its Tsirelson bound is $4\sqrt{5}-5$. It is the simplest Boole inequality that is not also a Bell inequality that can be violated by quantum mechanics. It was also the first such inequality to be discovered\footnote{Pitowsky found the inequalities for the $3$-cycle in 1989 \cite{pitowsky89}, but they can not be violated by quantum mechanics.}. Since this inequality can violated by qutrits, and only requires the measurement of 5 observables, it allows one of simplest possible tests of noncontextuality. Such an experimental test has in fact been carried out \cite{lapkiewicz11}.

	\section{Tsirelson bounds for Boole inequalities}\label{sec:cabelloseveriniwinter}

	In this section we consider the problem of calculating Tsirelson bounds for generic Boole inequalities. In general, this is extremely difficult to do. The best known algorithm for solving it involves an infinite hierarchy of semidefinite programs \cite{navascues07,navascues08,fritz12}, with each step providing a tighter upper bound to the Tsirelson bound. This algorithm, however, does not terminate, since it can never confirm that a given upper bound is in fact equal to the Tsirelson bound. For this and other reasons, Tsirelson bounds are conjectured to be in general uncomputable  \cite{fritz12,wolf11,fritz12b} (to the best of my knowledge, this was first conjectured by Tobias Fritz). 
	
	But since the computation of each step of the hierarchy is a semidefinite program, it can be done efficiently, and in practice good upper bounds can be obtained with little effort. Here we present a simple technique to find an upper bound, due to Cabello, Severini, and Winter \cite{cabello10}, that is closely related to the first step of the hierarchy \cite{fritz12}.

	To study the quantum value of a Boole inequality $\chsh$ it is more convenient to represent it as an operator; \ie, we define $\hat\chsh$ to be the operator such that $\chsh = \mean{\hat\chsh}_\rho$. For example, for the $n$-cycle inequalities
 	\[
	\chsh_n = \sum_{i=0}^{n-1} \gamma_i \mean{X_iX_{i+1}} \le n-2,
	\]
	the operator is
 	\[
	\hat\chsh_n = \sum_{i=0}^{n-1} \gamma_i X_iX_{i+1}
	\]
	Then we define the Tsirelson bound $\Omega_\text{Q}$ of a Boole inequality $\chsh$ from some marginal scenario $\marginal$ as
	\begin{equation}\label{eq:maxquantum}
	\Omega_\text{Q} =  \max_{\rho,\hat\chsh} \tr\rho\hat\chsh = \max_{\hat\chsh}\norm{\hat\chsh},
	\end{equation}
	where\footnote{$\norm{\cdot}$ is the standard operator norm.} the maximization is done over all quantum realizations $\hat\chsh$ of the marginal scenario $\marginal$, \ie, over all operators $X_i$ which respect the commutation relations implied by the marginal scenario. Note that the Tsirelson bound is always reached with pure states.
	
	The problem is that doing this maximization is a terribly difficult job, as the set of quantum realizations of a marginal scenario is anything but simple. It is possible, however, to do the maximization over a larger set, and thus obtain an upper bound on $\Omega_\text{Q}$. To show how to do this, we need a little detour through graph theory.
	
	\subsection{A graph-theoretical detour}	
	
	To be more precise, we are going to show that every quantum realization of a marginal scenario is also an orthonormal representation of a certain graph, and that the maximization over these orthonormal representations can be done efficiently.
	
	This graph is called the \emph{exclusivity} graph of a given Boole inequality. Please do not confound it with the compatibility graph\footnote{Actually, in the general case it is the compatibility hypergraph: it is only a graph when the maximum number of observables in a context is two.} that was introduced in the previous section; the compatibility graph encodes the marginal scenario. The exclusivity graph, on the other hand, encodes a specific representation of a Boole inequality.
	
	To define it, we first need to rewrite the desired Boole inequality as the conical sum
	\begin{equation}\label{eq:boolecsw}
	\Sigma = \sum_{c_i} \omega_i p(c_i|C_i),
	\end{equation}
	that is, as a sum of probabilities with positive coefficients\footnote{It can be proven that these coefficients are always rational.}. This can always be done, since we can always eliminate negative signs through normalization of probabilities, \ie, using the fact that $-p(A=a) = p(A\neq a)-1$. A useful identity for doing this with inequalities that are originally written in terms of expectation values is
	\begin{equation}\label{eq:cswtransform}
	\pm\mean{X_iX_j} = 2\bigg(p(+\pm|X_i,X_j)+p(-\mp|X_i,X_j)\bigg)-1.
	\end{equation}
	It is easy to see that this representation is not unique, because the representation of the Boole inequalities themselves is not unique: there is freedom in using no-disturbance conditions (or no-signalling) and normalization. For example, the inequality \eqref{eq:osspekkens} for the $3$-cycle can be represented as
	\begin{multline}\label{eq:3ciclocsw}
	p(+-|01) + p(-+|01) + p(+-|12) + p(-+|12) \\+ p(+-|20) + p(-+|20) \le 2,
	\end{multline}
	which can be further simplified through no-disturbance conditions to
	\[
	p(+-|01) + p(+-|12) + p(+-|20) \le 1.
	\]
	For clarity, we are omitting the $X$s from these inequalities. In general, different representations will give you different upper bounds for the Tsirelson bound\footnote{Although this is not the case with this example.}, and it is a bit of an art to find the best representation \cite{sadiq11}. 
	

	Now, we're ready to define the exclusivity graph:
	\begin{definition}
	The exclusivity graph of a Boole inequality written in the form \eqref{eq:boolecsw} is a graph that has the events $c_i|C_i$ as vertices, with edges connecting exclusive events.
	\end{definition}
	In the general case, we need to consider the exclusivity graph together with an assignment positive numbers $\omega_i$ to its vertices. It is often the case that $\omega_i = 1$, and then we don't need to talk about this.
	
	For example, the vertex $+-|01$ from inequality \eqref{eq:3ciclocsw} will be connected to the vertices $+-|12$, $-+|01$, and $+-|20$. Its exclusivity graph is the prism graph represented in figure \ref{fig:3ciclo}.

	\begin{figure}[ht]
	\centering
	\includegraphics[width=0.5\textwidth]{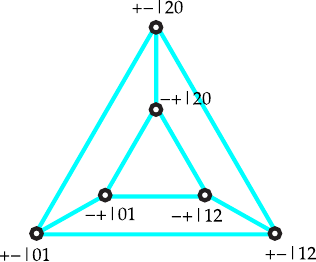}
	\caption{Exclusivity graph for the $3$-cycle.}
	\label{fig:3ciclo}
	\end{figure}

	Now we're ready to define what is an orthonormal representation\footnote{Unfortunately, our definition is the opposite of what is found in the graph theory literature \cite{lovasz79,knuth94}: what they call an orthonormal representation of a graph $G$ is equivalent to our definition of an orthonormal representation of the complement graph $\bar G$.}  of a graph:
	\begin{definition}
	An orthonormal representation of a graph $G$ with vertices $V_i$ is an assignment of projectors $\Gamma_i$ such that $V_i$ adjacent to $V_j$ implies that $\Gamma_i\Gamma_j=0$.
	\end{definition}
	
	Now we can show that any quantum realization of a marginal scenario gives rise to an orthonormal representation of the exclusivity graph. The idea is quite simple: in quantum mechanics, to obtain the probability $p(c_i|C_i)$ one calculates $\tr \rho \Pi_{c_i|C_i} $; then if we define $\Gamma_i = \Pi_{c_i|C_i}$, this is a valid orthonormal representation, as projectors associated to exclusive events are orthogonal.
	
	For example, in figure \ref{fig:3ciclo} the vertices $+-|01$ and $+-|20$ are adjacent. To calculate the probabilities we have 
	\[ p(+-|01) = \tr \rho \Pi_0^+\Pi_1^-\mathand p(+-|20) = \tr\rho \Pi_2^+\Pi_0^-,\]
	where $\Pi_0^+$ is the projector of the observable $X_0$ associated with the outcome $+$, and so on. Therefore, the projectors we assign to these vertices are $\Pi_0^+\Pi_1^-$ and $\Pi_2^+\Pi_0^-$, and their product is zero since they $\Pi_0^+\Pi_0^- = 0$.
	
	The converse statement is not true: given an orthonormal representation $\Gamma_i$ of the graph, it is in general not possible to find a quantum realization of the associated marginal scenario such that $\Pi_{c_i|C_i}=\Gamma_i$.

	An elegant counterexample can be found by considering the CHSH\footnote{Actually, this relabelling of it: $-\mean{X_0X_1} - \mean{X_1X_2}+\mean{X_2X_3}-\mean{X_3X_0} \le 2$} \eqref{eq:chsh} and Klyachko \eqref{eq:klyachko} inequalities \cite{sadiq11}. Using normalization and no-disturbance conditions they can be written, respectively, as 
	\begin{equation}\label{eq:sadiq}
	p({+-}|01)+p({+-}|12)+p({++}|23)+p({-+}|30)+p({-}|0) \le 2 \le \frac{3+\sqrt2}2
	\end{equation}
	and
	\begin{equation}\label{eq:klyachkoboole}
	p({+-}|01)+p({+-}|12)+p({+-}|23)+p({+-}|34)+p({+-}|40) \le 2 \le \sqrt5,
	\end{equation}
	where the last inequalities refer to the respective Tsirelson bounds. The surprising thing about these inequalities is that their exclusivity graph is the \emph{same}, the pentagon. They are shown in figure \ref{fig:sadiq}. 

	\begin{figure}[ht]
	\centering
	\includegraphics[width=0.35\textwidth]{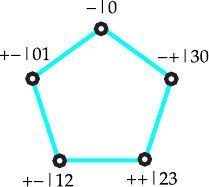} \hspace*{25pt} \includegraphics[width=0.35\textwidth]{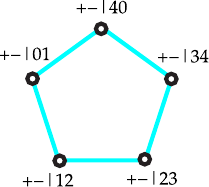}
	\caption{Exclusivity graphs for inequalities \eqref{eq:sadiq} and \eqref{eq:klyachkoboole}.}
	\label{fig:sadiq}
	\end{figure}
	
	The fact that the Tsirelson bound for the Klyachko inequality \eqref{eq:klyachkoboole} is $\sqrt{5}$ implies that there is an orthonormal representation of the pentagon $\Gamma_i$ such that $\norm{\sum_i \Gamma_i} = \sqrt{5}$; but if there existed a quantum realization $\Pi_{c_i|C_i}$ of the CHSH scenario such that $\Pi_{c_i|C_i} = \Gamma_i$, this would imply that it is possible to violate inequality \eqref{eq:sadiq} up to $\sqrt{5}$, a contradiction, since $\sqrt{5}$ is larger than its Tsirelson bound of $\frac{3+\sqrt2}2$.
	
	This shows that the set of orthonormal representations of an exclusivity graph is strictly larger than the set of quantum realizations of a marginal scenario, and therefore that optimizing over this larger set can only give us an upper bound on the Tsirelson bound of the Boole inequality.
	
	\subsection{The CSW theorem}

	Putting all these observations together shows us that
	\[
	\Omega_\text{Q} = \max_{\psi,\Pi_{c_i|C_i}} \sum_i \omega_i\tr\psi\Pi_{c_i|C_i} \le \max_{\psi,\Gamma_i} \sum_i \omega_i\tr\psi \Gamma_i = \Omega_\text{Q}',
	\]
	that is, the Tsirelson bound $\Omega_\text{Q}$ of a Boole inequality written in the form \eqref{eq:boolecsw} is upperbounded by maximizing the value of the inequality over all orthonormal representations $\Gamma_i$ of its exclusivity graph $G$.
	
	The significance of this observation comes from the fact that while $\Omega_Q$ is in general uncomputable, $\Omega_Q'$ can be calculated in polynomial time. Furthermore, 
	\[\Omega_Q' = \vartheta(G,\omega),\]
	that is, it is the weighted Lovász $\vartheta$-function of the exclusivity graph \cite{lovasz79,knuth94}. The proof of this equivalence\footnote{Note, once more, that Lovász's definition of an orthonormal representation of a graph $G$ is equivalent to our definition of an orthonormal representation of the complement graph $\bar G$.} for $\omega_i=1$ is theorem 5 in \cite{lovasz79}, or equation 10.1 in \cite{knuth94} for the general case. This is a famous graph-theoretical function, and therefore this equivalence allows us to access the vast literature existent about it to find Tsirelson bounds for the inequalities that interests us.
	
	There's one caveat: the usual definition of the Lovász function requires the $\Gamma_i$ to be one-dimensional projectors, and in our case this is not always true, since if we have $\Gamma_i=\Pi_{c_i|C_i}$ they will have in general rank larger than one. But this poses no problem, since restricting the maximization to be over one-dimensional projectors does not reduce the value of $\Omega_Q'$. To see that, suppose that the maximum is reached with a higher-dimensional orthonormal representation $\Gamma_i$. Then we can simply define define one-dimensional projectors 
	\[\Gamma_i' = \frac{\Gamma_i\psi\Gamma_i}{\tr\psi\Gamma_i} \]
	such that $\Gamma_i'$ is also an orthonormal representation of $G$ and $\tr\psi\Gamma_i' = \tr\psi\Gamma_i$, thereby giving the same value of $\Omega_Q'$.

	This connection of Boole inequalities with graph theory is known as the CSW theorem:
	
	\begin{theorem}[Cabello, Severini, Winter \cite{cabello10}]\label{teo:csw}
	Let $ \Sigma = \sum_{c_i} \omega_i p(c_i|C_i)$
	be a Boole inequality, $G$ its exclusivity graph, and $\Omega_Q$ its Tsirelson bound. Then 
	\begin{equation}\label{eq:maxlovasz}
	\Sigma \le \Omega_\text{Q} \le  \vartheta(G,\omega)
	\end{equation}
	\end{theorem}

	\subsection{\texorpdfstring{Tsirelson bounds for the $n$-cycle}{Tsirelson bounds for the n-cycle}}

	As an application of CSW theorem \ref{teo:csw}, we shall find the quantum bounds for the Boole inequalities found in section \ref{sec:nciclo}. As these inequalities only have terms $\pm\mean{X_iX_j}$, the transformation \eqref{eq:cswtransform} will be enough to bring them to the form of inequality \eqref{eq:boolecsw}, so
	\[ \chsh_n = 2\Sigma_n - n,\]
	where $\Sigma_n$ is the desired sum of probabilities. To find the exclusivity graph for odd $n$, the same strategy used in figure \ref{fig:3ciclo} works, so it will be the prism graph $Y_n$, and therefore the Tsirelson bound is upperbounded by $2\vartheta(Y_n)-n$. The Lovász function of the prism graph is\footnote{As was proven in \cite{araujo12}, and can also be derived from the results of \cite{liang11,cabello10}.}
	\[ \vartheta(Y_n) = \frac{2n \cos\frac\pi n}{1+\cos\frac\pi n}, \]
	thus proving that the quantum violation \eqref{eq:oddviolation} is the largest possible.

	To find the exclusivity graph for even $n$, the strategy is as represented in figure \ref{fig:4ciclo}, where it is done for $n=4$. It is clear that this strategy always works, so the exclusivity graph for even $n$ is the Möbius ladder $M_{2n}$. Its Lovász function is conjectured to be\footnote{See \cite{araujo12} for a discussion.}
	\[ \vartheta(M_{2n}) = \frac{n}{2}\de{1+\cos\frac\pi n}, \]
	which would prove that the quantum violation \eqref{eq:evenviolation} is in fact the largest possible. A proof can be obtained\footnote{See, again, \cite{araujo12} for a discussion.} from the results of \cite{wehner06}.

	\begin{figure}[ht]
	\centering
	\includegraphics[width=0.50\textwidth]{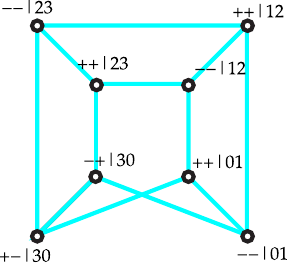}
	\caption{CSW graph for the $4$-cycle.}
	\label{fig:4ciclo}
	\end{figure}
	
	\section{State-independent Boole inequalities}

	All the Boole inequalities we have studied so far have quantum violations that depend on the quantum state: they are violated by some, but not violated by others. This situation stands in contrast with the proofs of contextuality we studied in section \ref{sec:kochenspecker}: they only considered predictions of quantum mechanics that were valid for \emph{any} state. Therefore, it would be quite surprising if we couldn't find a Boole inequality that were violated by \emph{any} quantum state.

\subsection{A Boole inequality from the 18-projector proof by Cabello, Estebaranz, and García-Alcaine}

	The 18-projector proof \cite{cabello96b} translates quite directly into a state-independent Boole inequality \cite{cabello08}. To see that, define $A_{ij} = 2v_{ij}-\id$, where $v_{ij}$ are the projectors from figure \ref{fig:adan18}. Then if we take the product of four commuting such $A_{ij}$, it will be always equal to $-\id$. Taking these products over all nine sets of commuting $A_{ij}$ and adding them together, we get
	\begin{multline}
	\hat{\mathcal{I}}_{18} = - A_{12} A_{16} A_{17} A_{18} - A_{12} A_{23}
	A_{28} A_{29} -  A_{23} A_{34} A_{37} A_{39} 
	 \\ -  A_{34} A_{45} A_{47} A_{48}  -
	 A_{45} A_{56} A_{58} A_{59}  -  A_{16} A_{56}
	A_{67} A_{69}  \\ -  A_{17} A_{37} A_{47}
	A_{67}  - A_{18} A_{28} A_{48} A_{58} 
	 - A_{29} A_{39} A_{59} A_{69} = 9\id
	\end{multline}
	but a computer program can easily check that in any noncontextual theory \begin{multline}
\mathcal{I}_{18} = -\langle A_{12} A_{16} A_{17} A_{18} \rangle-\langle A_{12} A_{23}
A_{28} A_{29}\rangle - \langle A_{23} A_{34} A_{37} A_{39} \rangle
 \\ - \langle A_{34} A_{45} A_{47} A_{48} \rangle -
\langle A_{45} A_{56} A_{58} A_{59} \rangle - \langle A_{16} A_{56}
A_{67} A_{69} \rangle \\ - \langle A_{17} A_{37} A_{47}
A_{67} \rangle -\langle A_{18} A_{28} A_{48} A_{58} \rangle
 - \langle A_{29} A_{39} A_{59} A_{69} \rangle \le 7.
\end{multline}
	This Boole inequality is therefore violated by \emph{any} quantum state.

\subsection{A Boole inequality from Yu and Oh's 13-projector proof}

	The projectors from Yu and Oh's 13-projector proof can also be used to form such a state-independent inequality \cite{yu12}, but their inequality is not a facet of the noncontextual polytope, and according to our definition not a Boole inequality at all. Fortunately, there \emph{is} a Boole inequality associated to their projectors, found by Cabello \etal \cite{cabello12}. It reads
	\begin{multline*}
	\mathcal{I}_{YO} = 2\mean{H_0} + \sum_{i=1}^3 \mean{Z_i} + \mean{Y_i^+} + \mean{Y_i^-} +2\mean{H_i} \\
	+\sum_{j=1}^3 \mean{Z_jY_j^+}+\mean{Y_j^+Y_j^-}+\mean{Y_j^-Z_j} \\
	-3\sum_{k=1}^3 \mean{Z_kY_k^+Y_k^-} -\sum_{C_i \in \,\mathcal C_2}\mean{C_i} \le 25,
	\end{multline*}
	where $Z_i = \id -2z_i$, $Y_i^\pm = \id -2y_i^\pm$, $H_i = \id -2h_i$, as defined in section \ref{sec:yuoh}, and $\mathcal C_2$ is the subset of two-observable contexts of Yu and Oh's marginal scenario. The operator $\hat{\mathcal{I}}_{YO} = (25+8/3)\id$ is again proportional to identity, and this inequality is the one in Yu and Oh's noncontextual polytope with the largest violation. As this inequality was found by a computer program we feel no need of reproducing a proof here.

\subsection{A Boole inequality from the Peres-Mermin square}

	Peres-Mermin's proof can also be adapted into such an inequality. Let $A_{ij}$ be the observables of the Peres-Mermin square as defined in equation \eqref{eq:peresmerminsquare}. Then it follows that
	\begin{multline*}
	\hat{\mathcal{I}}_{PM} = A_{11} A_{12} A_{13} + A_{21} A_{22} A_{23} + A_{31} A_{32} A_{33} \\
	+ A_{11} A_{21} A_{31} +  A_{12} A_{22} A_{32} -  A_{13} A_{23} A_{33} = 6\id,
	\end{multline*}
	but a computer program\footnote{Or in fact yourself, by some playing around with the triangle inequality.} can easily check that
	\begin{multline*}
	\mathcal{I}_{PM} = \langle A_{11} A_{12} A_{13}\rangle + \langle A_{21} A_{22} A_{23}\rangle + \langle A_{31} A_{32} A_{33}\rangle \\+ \langle A_{11} A_{21} A_{31}\rangle + \langle A_{12} A_{22} A_{32}\rangle - \langle A_{13} A_{23} A_{33}\rangle \le 4.
	\end{multline*}
	This Boole inequality was also found by Adán Cabello \cite{cabello08}.

	Note that in all these inequalities the operator $\hat{\mathcal{I}}$ was proportional to identity, but this is not a required condition for a state-independent violation: we only need $\mean{\hat{\mathcal{I}}}_\psi$ to be larger than the noncontextual bound for every $\psi$. It is an open question if there is a Boole inequality that satisfies the latter condition but not the former\footnote{It is trivial, however, to generate such inequalities that are not facets of the noncontextual polytope.}.

\chapter*{Conclusion}
\addcontentsline{toc}{chapter}{Conclusion}

	The attentive reader might have noticed that despite hints of \textit{quantum magic} as the motivation for this thesis, there has been almost no mention of it in the technical parts of the text. In part this is because of the limitations of time and space, but more importantly because I believe that to really understand \textit{quantum magic}, we must understand the foundations of quantum mechanics first; and this latter understanding is still sorely lacking. The goal of this thesis was therefore to help with this point.

	This goal can be naturally split in two parts (if not in two chapters): first, to summarize old research in a clear and consistent way, and second (and more important), to expose new research that is not as widely known as I think it deserves to be.

	Specifically, I hope to have convinced the reader that the formulation of noncontextuality exposed in chapter \ref{cha:prob} is a fruitful way of separating ``classical'' phenomena from those that are truly quantum.  The way ahead is to actually pick up those fruits: develop information processing protocols that derive their strenght from the violation of Boole inequalities. In a sense, this work has already begun: we know that the higher-than-classical power of quantum random access codes comes from contextuality \cite{galvao02}, and \cite{liang11} has a very colourful description of a game in which contextuality boosts the chance of success.

	But, in my opinion, these protocols lack a deeper appeal, since it's not clear if the fact that they have a quantum advantage means anything other than the fact that they have a quantum advantage. What would really please me is to find a connection between contextuality and a discovery that has far-reaching implications in physics, mathematics, and computer science: quantum computing.

\appendix

\chapter{The Bell-Mermin model}\label{sec:bellmermin}

	This ontological model was first proposed by Bell in 1964 \cite{bell66}, in order to provide a counterexample to von Neumann's theorem \cite{vonneumann32}, and later cleaned up by David Mermin \cite{mermin93}. It is certainly the simplest deterministic ontological model out there, having been constructed to describe the statistics coming from the measurement of any observable of a pure qubit. It is not contextual, but if extended to mixed states it would have to be preparation-contextual, by Spekkens' theorem, and if extended to higher dimensions it would become measurement-contextual, by Gleason's theorem. It also can't be extended to describe POVMs, by Busch's theorem. In a sense, then, it is the best that a realist commited to non-contextuality can do.

	This model is quite out of fashion, as it measures observables instead of its projectors; but we shall make no violence to it by ``fixing'' this feature. The concerned reader may do it himself quite easily, or simply consult Harrigan's work \cite{harrigan10}. 

	We formulate it by representing a two-dimensional self-adjoint observable $A$ in the Bloch basis, as \[A = a_0 \id + a \cdot \sigma,\] where $a_0 \in \R$, $a \in \R^3$ and $\sigma$ is the vector of Pauli matrices. 

	The ontic space $\Lambda = S^2 \times S^2$ is the cartesian product of two unit spheres. In the first one we shall embed the pure states via their Bloch vector $\hat{\psi} \in S^2$, defined by $\psi = \frac{1}{2}(\id + \hat{\psi}\cdot \sigma)$, and in the second one we shall use an auxiliar unit vector $\lambda$.

	The ontic state is then
	\[ \mu_\psi = \delta(\lambda_\psi-\hat{\psi}), \]
	and the response function is
	\begin{equation}\label{eq:bellmermin}
	\xi_{A}(\lambda_\psi,\lambda) = a_0 + \norm{a} \sign(a\cdot(\lambda + \lambda_\psi)).
	\end{equation}
	Notice that given $\lambda_\psi$ and $\lambda$, it gives deterministically $a_0 + \norm{a}$ or $a_0 - \norm{a}$, as required.

	To recover the quantum statistics, we take the uniform average of $\xi_{A}$ over $\Lambda$:
	\begin{align*} 
\mean{A}  &= \int_\Lambda \mu_\psi \xi_{A} \\
		   &= a_0 + \norm{a}\int_\Lambda \dint \lambda_\psi \dint \lambda \, \delta(\lambda_\psi -\hat{\psi})\sign(\hat{a}\cdot(\lambda + \lambda_\psi)) \\
		   &= a_0 + \norm{a}\int_{S^2} \dint \lambda\, \sign(\hat{a}\cdot(\lambda + \hat{\psi})) \\
		   &= a_0 + \norm{a}\frac{1}{4\pi}\int_0^{2\pi}\int_0^\pi \sin\theta_{a\lambda}\dint \theta_{a\lambda}\dint\varphi\, \sign(\cos \theta_{a\lambda} + \hat{a}\cdot\hat{\psi}) \\
		   &= a_0 + \norm{a}\frac{1}{2}\de{\int_0^{\cos^{-1}(-\hat{a}\cdot\hat{\psi})}\sin\theta_{a\lambda}\dint \theta_{a\lambda}-\int_{\cos^{-1}(-\hat{a}\cdot\hat{\psi})}^\pi\sin\theta_{a\lambda}\dint \theta_{a\lambda}} \\
		   &= a_0 + a\cdot\hat{\psi} \\
		   &= \tr A \psi
	\end{align*}

\bibmark
\printbibliography 
\end{document}